\documentclass[a4paper,12pt]{article}
\pdfoutput=1
\usepackage{test,booktabs,graphicx,caption}

\usepackage{hyperref}
\hypersetup{colorlinks,citecolor=blue,urlcolor=magenta}
\usepackage{doi}

\usepackage[style=ext-authoryear-comp,
sorting=nyt,
dashed=false, 
maxcitenames=2, 
maxbibnames=99, 
uniquelist=false,
uniquename=false,
useprefix=true,
giveninits=true, 
natbib, 
date=year 
]{biblatex}

\AtBeginRefsection{\GenRefcontextData{sorting=ynt}}
\AtEveryCite{\localrefcontext[sorting=ynt]}

\DeclareFieldFormat{pages}{#1} 
\renewbibmacro{in:}{\ifentrytype{article}{}{\printtext{\bibstring{in}\intitlepunct}}} 

\DeclareFieldFormat[article,inbook,incollection,inproceedings,patent,thesis,unpublished]{titlecase:title}{\MakeSentenceCase*{#1}} 

\usepackage[inline,shortlabels]{enumitem}
\setlist[enumerate,1]{label=(\roman*)}

\usepackage[margin = 1.25in]{geometry}
\usepackage{setspace}
\newcommand{\cE}{\mathcal{E}}
\newcommand{\cK}{\mathcal{K}}
\newcommand{\cP}{\mathcal{P}}

\title{Global Characterization of Equilibria in Tirole's (1985) Model with a Dividend-Paying Asset\footnote{This paper is a revised version of our earlier paper ``Long-Run Behavior of Equilibrium in Tirole (1985)'s Model with Dividend-Paying Asset'' \citep{PhamTodaTirole} posted to arXiv under a different title in January 2025. We thank Stefano Bosi 
and Cuong Le Van for constructive comments and suggestions. Toda acknowledges financial support from Japan Center for Economic Research, Murata Science and Education Foundation, and Zengin Foundation.}}
\author{Ngoc-Sang Pham\thanks{EM Normandie Business School, M\'etis Lab. Email: \href{mailto:npham@em-normandie.fr}{npham@em-normandie.fr}.} \and Alexis Akira Toda\thanks{Department of Economics, Emory University and Research Institute for Economics and Business Administration, Kobe University. Email: \href{mailto:alexis.akira.toda@emory.edu}{alexis.akira.toda@emory.edu}.}}

\numberwithin{equation}{section}
\numberwithin{lem}{section}
\numberwithin{prop}{section}

\begin{document}
\maketitle

\begin{abstract}
We revisit Tirole's classic paper ``Asset Bubbles and Overlapping Generations'' (1985, \emph{Econometrica}) in the case of a dividend-paying asset. Recently, Pham and Toda (2026) constructed a counterexample to Proposition 1(c), showing that Tirole's equilibrium classification is incorrect as stated and that long-run outcomes can depend on initial capital. This paper characterizes the entire set of equilibrium initial asset prices under capital over-accumulation. Exactly one of three regimes occurs: (i) a unique bubbleless equilibrium with capital converging to zero (capital collapse), (ii) a unique asymptotically bubbly equilibrium converging to a positive steady state (bubble necessity), or (iii) a continuum of equilibria with different long-run bubble behavior (indeterminacy). We further derive a threshold for initial capital under the bubble necessity condition, establish preference-free sufficient conditions for capital collapse, and show that the continuum in the pure bubble model survives sufficiently small dividend perturbations. Closed-form examples illustrate the possible long-run outcomes.

\medskip

		
\medskip

\end{abstract}
	
\section{Introduction}

In a seminal paper, \citet{Tirole1985} showed that asset price bubbles can
arise in an overlapping generations (OLG) economy with production and can solve the well-known capital over-accumulation problem by absorbing
saving that would otherwise finance physical capital. Beyond this
central insight, Tirole offered a sharp classification of equilibrium
multiplicity and long-run behavior in terms of growth and interest
rates. This combination of economic intuition and tractable predictions
has made the model a benchmark for studying rational bubbles in
production economies.

More precisely, Proposition 1 of \citet{Tirole1985} can be summarized as
follows. Consider \citet{Diamond1965}'s overlapping generations neoclassical
growth model and introduce a dividend-paying asset in fixed supply. An equilibrium is called \emph{bubbleless} if the asset price ($P$) equals the fundamental value of the asset ($V$), defined by the present discounted value of dividends, and \emph{bubbly} otherwise ($P>V$). To distinguish long-run bubble size, an equilibrium is called \emph{asymptotically bubbleless} if the bubble per capita has lower limit zero and \emph{asymptotically bubbly} if it remains bounded away from zero. Let $G$ be the economic (population) growth rate, $G_d$ the dividend growth rate, and $R$ the steady-state interest rate in the absence of the asset. Under a monotonicity condition ensuring that saving is increasing in the interest rate, Proposition 1 of \citet{Tirole1985} claims that there are three cases depending on the magnitude of $G,G_d,R$:
\begin{enumerate}[(a)]
    \item If $R>G$, there exists a unique equilibrium, it is bubbleless, and the interest rate converges to $R$.
    \item If $G_d<R<G$, there exists a continuum of equilibria with the initial asset price in some interval $p_0\in [\ubar{p}_0,\bar{p}_0]$. Furthermore, $p_0=\ubar{p}_0$ is bubbleless, any $p_0\in (\ubar{p}_0,\bar{p}_0)$ is bubbly but asymptotically bubbleless, and $p_0=\bar{p}_0$ is asymptotically bubbly.
    \item If $R<G_d<G$, there exists a unique equilibrium, it is asymptotically bubbly, and the interest rate converges to $G$.
\end{enumerate}
\citet{Tirole1985} provides strong intuition for these results, and the special case with zero dividends (the so-called ``pure bubble'' case) has been extended to a variety of settings.

Recently, \citet{PhamToda2026ECMA} constructed a counterexample to Proposition 1(c) of \citet{Tirole1985}, showing that both its statement and proof are incorrect. Their example has a unique equilibrium in which capital converges to zero even though the assumptions of Proposition 1(c) are satisfied. They also restore its conclusion under the additional conditions that initial capital is sufficiently large and dividends are sufficiently small. Thus, the long-run outcome depends not only on the growth-rate comparison in Tirole's proposition but also on the initial condition. This counterexample establishes that Tirole's equilibrium classification is incorrect as stated. Given the significance of Tirole's paper, a global characterization of equilibria in this classic model is therefore warranted---not only a correction of Proposition 1(c), but a reexamination of Proposition 1 as a whole.

This paper provides such a characterization for \citet{Tirole1985}'s model with general, potentially nonstationary dividend processes and arbitrary initial capital. Theorem \ref{thm:eqset} characterizes all possible forms of the equilibrium set under capital over-accumulation. The contrast with the Diamond model is useful. Without the asset, the dynamics are one-dimensional, so the initial capital stock uniquely determines the equilibrium path. With the asset, the initial price $p_0$ is a free variable and the dynamics become two-dimensional. The theorem shows that exactly one of three regimes occurs:
\begin{enumerate*}
    \item \emph{capital collapse}: a unique bubbleless equilibrium in which capital converges to zero and the interest rate diverges;
    \item \emph{bubble necessity}: a unique asymptotically bubbly equilibrium converging to a positive bubbly steady state; or
    \item \emph{indeterminacy}: a continuum of equilibria whose interior elements are bubbly but asymptotically bubbleless and whose upper endpoint is asymptotically bubbly.
\end{enumerate*}

The subsequent results provide guidance on which regime occurs. Under the bubble necessity condition $R<G_d<G$ \citep{HiranoToda2025JPE}, Theorem \ref{thm:necessity} rules out indeterminacy and establishes uniqueness. The unique equilibrium displays capital collapse when initial capital is below a threshold and is asymptotically bubbly when initial capital is above it. Theorem \ref{thm:curse} gives explicit, preference-free sufficient conditions for capital collapse when initial capital is small: the wage function must have a power-law behavior near $k=0$ and detrended dividends must lie between suitable geometric bounds. Unlike capital collapse that may already arise in the Diamond economy under weak Inada conditions \citep{GalorRyder1989}, Theorem \ref{thm:curse} applies to general technologies including Cobb-Douglas for which the Diamond economy globally converges to a positive steady state. The collapse is therefore caused by the dividend-paying asset's crowding out of capital investment. Theorem \ref{thm:continuum} provides a sufficient condition for indeterminacy. If initial capital is at least the golden rule value and the present discounted value of dividends in the \citet{Diamond1965} economy does not exceed the steady-state bubble, then the equilibrium set is a nondegenerate interval and its minimum corresponds to the unique bubbleless equilibrium. In particular, the continuum of pure bubble equilibria in \citet{Tirole1985} survives sufficiently small dividend perturbations. Indeterminacy is familiar in overlapping generations pure bubble models since \citet{Gale1973}; what Theorem \ref{thm:continuum} adds is a condition stated entirely in terms of primitives including dividends. The ordering arguments underlying Theorems \ref{thm:bubbleless}--\ref{thm:necessity} and \ref{thm:continuum} require saving to be increasing in the interest rate, as imposed in Assumption \ref{asmp:s}; the preference-free capital collapse result in Theorem \ref{thm:curse} does not.

\S\ref{sec:example} complements the general results with closed-form examples. Under log utility and Cobb-Douglas production, Lemma \ref{lem:CD} transforms the equilibrium system so that one can choose an auxiliary sequence and recover the corresponding dividend process. Examples \ref{exmp:k0}--\ref{exmp:asym_bubbleless} construct, respectively, a unique bubbleless equilibrium with capital collapse, a unique asymptotically bubbly equilibrium, and a bubbly but asymptotically bubbleless equilibrium. Taken together, the theoretical and analytical results show why each part of Proposition 1 of \citet{Tirole1985} requires qualification: in each case, capital collapse is possible and the initial capital stock and the size of dividends matter.

Our paper belongs to the so-called ``rational bubble'' literature, in which the asset price can exceed the present discounted value of dividends in general equilibrium with rational agents. Foundational contributions include \citet{Samuelson1958}, \citet{Bewley1980}, \citet{Tirole1985}, \citet{Kocherlakota1992}, and \citet{SantosWoodford1997}. Subsequent work studies, among other issues, endogenous growth, financial frictions, liquidity, and equilibrium determinacy; representative contributions include \citet{Olivier2000}, \citet{CaballeroKrishnamurthy2006}, \citet{HellwigLorenzoni2009}, \citet{FarhiTirole2012}, \citet{MartinVentura2012}, \citet{HiranoYanagawa2017}, \citet{BloiseCitanna2019}, \citet*{BosiLeVanPham2022}, and \citet{Sorger2026}. See \citet{MartinVentura2018} and \citet{HiranoToda2024JME,HiranoTodaReal} for reviews.

The closest antecedent is \citet{PhamToda2026ECMA}. That paper isolates the failure of Proposition 1(c) of \citet{Tirole1985}, constructs a counterexample, and restores the conclusion when initial capital is large and dividends are small. The present paper addresses a different question: it characterizes the full equilibrium set under capital over-accumulation, identifies the three possible global regimes, provides an explicit sufficient condition for the continuum of equilibria, and develops closed-form examples of the distinct long-run outcomes.\footnote{\citet*{BosiLeVanPham2025} study similar issues in an OLG exchange economy with time-dependent endowments.} Thus, the substantive contribution here is not the counterexample already established in \citet{PhamToda2026ECMA}, but the characterization of the equilibrium set and the conditions governing capital collapse, bubble necessity, and indeterminacy.

Two other closely related papers are \citet*{BosiHa-HuyLeVanPhamPham2018} and \citet{HiranoToda2025JPE}. \citet{BosiHa-HuyLeVanPhamPham2018} extend \citet{Tirole1985}'s model with altruism and provide conditions under which
\begin{enumerate*}
    \item there is no bubble or
    \item there exists a continuum of bubbly equilibria,
\end{enumerate*}
which are questions addressed in Proposition 1(a)(b) of \citet{Tirole1985}. \citet{HiranoToda2025JPE} provide conditions under which any equilibrium (if it exists) must have a bubble in macro-finance models, which is related to Proposition 1(c) of \citet{Tirole1985}. Relative to these papers, our characterization of the equilibrium set is sharper: for instance, \citet[\S V.A]{HiranoToda2025JPE} provide a conditional bubble necessity result in \citet{Tirole1985}'s model for equilibria that do not collapse, but they neither characterize the equilibrium set, identify its endpoints and intermediate equilibria, nor establish the global trichotomy in Theorem \ref{thm:eqset}, and utility is restricted to log. A detailed result-by-result comparison with \citet{Tirole1985} and these related papers is provided in Appendix \ref{sec:disc}.

\section{Model}\label{sec:model}

\subsection{\texorpdfstring{\citet{Tirole1985}'s model}{}}\label{subsec:model_tirole}
We briefly review \citet{Tirole1985}'s model, which introduces a long-lived asset to \citet{Diamond1965}'s model. Time is discrete and infinite, indexed by $t=0,1,\dotsc$. There is a homogeneous good whose spot price is normalized to 1.

\paragraph{Agents}
There are overlapping generations of agents that live for two periods, when young and old. Let $N_t>0$ be the population of the young at time $t$, which is exogenous. Each young agent is endowed with a unit of labor, which is supplied inelastically. The old do not have any labor endowment. Therefore, the aggregate labor supply at time $t$ is also $N_t$. Agents in generation $t$ have utility function $U_t(c_t^y,c_{t+1}^o)$, where $c_t^y,c_{t+1}^o>0$ denote the consumption of an agent in generation $t$ when young and old. We assume $U_t:\R_{++}^2\to \R$ is continuous, quasi-concave, and strictly increasing. The initial old only care about their consumption $c_0^o$.

\paragraph{Production}
At time $t$, a representative firm produces the good using capital (denoted $K_t$) and labor (denoted $L_t$) as inputs. Let $F_t(K_t,L_t)$ be the output, where $F_t$ is a neoclassical production function, meaning that $F_t:\R_+^2\to \R_+$ is homogeneous of degree 1, concave, and continuously differentiable on $\R_{++}^2$ with positive partial derivatives. Capital fully depreciates after production.\footnote{Full capital depreciation is without loss of generality. To see why, suppose output is $f(K,L)$ and capital depreciates at rate $\delta\in [0,1]$. Then the output including undepreciated capital is $F(K,L)=f(K,L)+(1-\delta)K$, and the analysis depends only on $F$. This point is obvious but noted in \citet[p.~1093]{Coleman1991}.} Letting $R_t>0$ be the capital rent and $w_t>0$ be the wage at time $t$ (which are both endogenous), at time $t$ the firm seeks to maximize profit
\begin{equation}
    F_t(K_t,L_t)-R_tK_t-w_tL_t. \label{eq:profit}
\end{equation}
Because $F_t$ is homogeneous of degree 1, if there is a solution to the profit maximization problem, the profit must be zero. Therefore, we do not need to specify the ownership of firms. The economy starts at $t=0$ with an exogenous stock of capital $K_0>0$, which is owned by the initial old.

\paragraph{Capital investment and asset}
The consumption good at time $t$ can be converted to capital available for production at time $t+1$ at a $1:1$ ratio. Thus converting one unit of the consumption good to capital at time $t$ yields the capital rent $R_{t+1}$ at time $t+1$. In addition to capital, there is a unit supply of a long-lived asset that pays exogenous dividend $D_t\ge 0$ at time $t$,\footnote{Alternatively, we may interpret the model as having one consumption good and three inputs for production---capital, labor, and land---and the time $t$ production function takes the form $F_t(K,L,X)=F(K,L)+D_tX$, where $X$ is the land input (which is in unit supply) and $D_t$ is land productivity.} which is initially owned by the old and can be freely disposed of. Let $P_t\ge 0$ be the price of the asset, which is endogenous.

\paragraph{Individual problem}

Agents maximize utility subject to the budget constraints, taking prices as given. For the initial old, noting that the population is $N_{-1}$ and the asset is in unit supply (so each initial old is endowed with capital $K_0/N_{-1}$ and $1/N_{-1}$ shares of the asset), the solution is
\begin{equation}
    c_0^o=\frac{F_0(K_0,L_0)-w_0L_0+P_0+D_0}{N_{-1}}. \label{eq:c0}
\end{equation}

An agent in generation $t\ge 0$ maximizes utility subject to the budget constraints
\begin{subequations}\label{eq:budget}
\begin{align}
    &\text{Young:} & c_t^y+i_t+P_tx_t&=w_t, \label{eq:budget_young}\\
    &\text{Old:} & c_{t+1}^o&=R_{t+1}i_t+(P_{t+1}+D_{t+1})x_t, \label{eq:budget_old}
\end{align}
\end{subequations}
where $i_t\ge 0$ denotes capital investment and $x_t\in \R$ denotes asset holdings.

\paragraph{Equilibrium}

Because the economy features no uncertainty, we focus on deterministic equilibria. Furthermore, because agents in each generation are homogeneous, without loss of generality we focus on symmetric equilibria in which each agent in the same cohort makes the same decision.

\begin{defn}\label{defn:eq}
Let initial capital $K_0>0$, young population $\set{N_t}_{t=-1}^\infty$, and dividends $\set{D_t}_{t=0}^\infty$ be given. A \emph{rational expectations equilibrium} consists of a nonnegative sequence
\begin{equation}
    \set{(P_t,R_{t+1},w_t,c_t^y,c_t^o,i_t,x_t,K_t,L_t)}_{t=0}^\infty \label{eq:seq1}
\end{equation}
such that the following conditions hold.
\begin{enumerate}
    \item\label{item:eq1} (Utility maximization) $c_0^o$ satisfies \eqref{eq:c0}; for each $t\ge 0$, $(c_t^y,c_{t+1}^o,i_t,x_t)$ maximizes utility subject to budget constraints \eqref{eq:budget}.
    \item\label{item:eq2} (Profit maximization) For each $t$, $(K_t,L_t)$ maximizes the profit \eqref{eq:profit}.
    \item\label{item:eq3} (Commodity market clearing) For each $t$, we have
    \begin{equation}
        N_t(c_t^y+i_t)+N_{t-1}c_t^o=F_t(K_t,L_t)+D_t. \label{eq:c_clear}
    \end{equation}
    \item\label{item:eq4} (Labor market clearing) For each $t$, we have $L_t=N_t$.
    \item\label{item:eq5} (Capital and asset market clearing) For each $t$, we have
    \begin{subequations}
    \begin{align}
        N_ti_t&=K_{t+1}, \label{eq:k_clear}\\
        N_tx_t&=1. \label{eq:x_clear}
    \end{align}
    \end{subequations}
\end{enumerate}
\end{defn}

The right-hand side of \eqref{eq:c_clear} is aggregate output at time $t$, which must be either consumed or invested as capital. (Recall that the consumption good can be converted to capital at a $1:1$ ratio.) The left-hand side of \eqref{eq:k_clear} is aggregate capital investment at time $t$, which must equal aggregate capital $K_{t+1}$. The left-hand side of \eqref{eq:x_clear} is aggregate asset holdings, which must equal 1 (because the asset is in unit supply).

Definition \ref{defn:eq} involves many objects. The following lemma simplifies the equilibrium conditions.

\begin{lem}\label{lem:eq}
A rational expectations equilibrium is equivalent to a nonnegative sequence $\set{(P_t,R_{t+1},w_t,s_t,K_t)}_{t=0}^\infty$ such that, for each $t$,
\begin{enumerate}
    \item (Utility maximization) saving $s_t$ solves
    \begin{equation}
        \max_{s\in [0,w_t]}U_t(w_t-s,R_{t+1}s), \label{eq:UMP}
    \end{equation}
    \item (Profit maximization) $(K_t,L_t)$ maximizes the profit \eqref{eq:profit},
    \item (No-arbitrage) we have
    \begin{equation}
        P_t=\frac{1}{R_{t+1}}(P_{t+1}+D_{t+1}), \label{eq:R}
    \end{equation}
    \item (Market clearing) we have $L_t=N_t$ and
    \begin{equation}
        N_ts_t=K_{t+1}+P_t. \label{eq:s_clear}
    \end{equation}
\end{enumerate}
\end{lem}

The left-hand side of \eqref{eq:s_clear} is aggregate savings at time $t$, which must be allocated to capital investment or asset purchase. (Recall that the asset is in unit supply.) The following proposition establishes the existence of equilibrium.

\begin{prop}\label{prop:exist}
A rational expectations equilibrium in Definition \ref{defn:eq} exists.
\end{prop}

\begin{rem}\label{rem:exist}
Proposition \ref{prop:exist} is a supporting result needed because the subsequent analysis allows arbitrary initial capital and nonstationary dividend processes. It establishes the nonemptiness of the equilibrium set used
in Proposition \ref{prop:p0} and throughout \S\ref{sec:main}, and its truncation argument is also used in the proof of Theorem \ref{thm:bubbleless} to construct a bubbleless equilibrium. \citet{Tirole1985} does not establish existence in this general setting, and the argument in his Lemma 1 is incomplete; see Appendix
\ref{sec:disc}. \citet*[Lemma 2]{BosiHa-HuyLeVanPhamPham2018} establish
existence in a generalized model with altruism, but impose the gross
substitutes property. In contrast, Proposition \ref{prop:exist} requires only standard conditions such as quasi-concavity.
\end{rem}

\begin{rem}\label{rem:multiple}
In our model description, we assume that there is only one dividend-paying asset. In contrast, \citet{Tirole1985} assumes that there is a dividend-paying asset whose price is always equal to its fundamental value and an intrinsically worthless asset that pays no dividends. The two approaches are equivalent. To see why, consider the two-asset setting in \citet{Tirole1985}, with $V_t$ the asset price equal to the fundamental value defined by the present discounted value of dividends
\begin{equation}
    V_t\coloneqq \sum_{s=1}^\infty \frac{D_{t+s}}{R_{t+1}\dotsb R_{t+s}} \label{eq:Vt1}
\end{equation}
and $B_t$ the bubble. The absence of arbitrage as in \eqref{eq:R} implies that
\begin{equation*}
    V_t=\frac{1}{R_{t+1}}(V_{t+1}+D_{t+1}) \quad \text{and} \quad B_t=\frac{1}{R_{t+1}}B_{t+1}.
\end{equation*}
Taking the sum and letting $P_t=V_t+B_t$, we obtain the no-arbitrage condition \eqref{eq:R}. Thus, \citet{Tirole1985}'s two-asset setting reduces to a one-asset setting with potentially an asset price bubble. Conversely, by starting with a one-asset setting like our model, we may define the fundamental value $V_t$ by \eqref{eq:Vt1} and the bubble $B_t\coloneqq P_t-V_t\ge 0$ to recover \citet{Tirole1985}'s two-asset setting.

More generally, if there are multiple long-lived assets, the absence of arbitrage implies that we can bundle them together as one asset. Thus, assuming a single dividend-paying asset is without loss of generality. See also \citet{HiranoTodaReal} for more discussion on this issue.
\end{rem}

\subsection{Equilibrium system}

In this section, we introduce additional assumptions to derive the dynamical system that describes the equilibrium. As we are interested in the long-run behavior of the equilibrium, we introduce the following stationarity assumption.

\begin{asmp}\label{asmp:G}
The utility function $U_t=U$ and the production function $F_t=F$ are time-invariant. Population is $N_t=G^t$, where $G>0$.
\end{asmp}

In what follows, let $(k_t,p_t,d_t)\coloneqq (K_t/N_t,P_t/N_t,D_t/N_t)$ be the detrended capital, asset price, and dividend.

\begin{asmp}\label{asmp:f}
The production function $F$ is homogeneous of degree 1, concave, and twice continuously differentiable on $\R_{++}^2$ with positive partial derivatives. Furthermore, $f(k)\coloneqq F(k,1)$ satisfies $f'(0)=\infty$, $f'(\infty)<G$, and $f''<0$.
\end{asmp}

Assumption \ref{asmp:f} is standard. The homogeneity of $F$ implies that the technology exhibits constant returns to scale. The condition $f'(0)=\infty$ is the Inada condition, which guarantees an interior solution. The condition $f'(\infty)<G$ prevents detrended capital from diverging to infinity. A typical example satisfying Assumption \ref{asmp:f} is the Cobb-Douglas production function
\begin{equation}
    F(K,L)=AK^\alpha L^{1-\alpha}+(1-\delta)K, \label{eq:CD}
\end{equation}
where $A>0$ is productivity, $\alpha\in (0,1)$ is output elasticity of capital, and $\delta\in [0,1]$ is the capital depreciation rate with $1-\delta<G$. Then $f(k)=Ak^\alpha+(1-\delta)k$ and $f'(\infty)=1-\delta<G$.

Under the maintained assumptions, the equilibrium is interior ($k_t>0$) and the detrended capital and asset price are uniformly bounded.

\begin{lem}\label{lem:k}
If Assumptions \ref{asmp:G}, \ref{asmp:f} hold, in equilibrium we have $k_t>0$, $R_t=f'(k_t)\ge f'(\infty)$, $w_t=f(k_t)-k_tf'(k_t)>0$, $k_{t+1}\le f(k_t)/G$, $p_t\le f(k_t)$, and $\sup_t k_t<\infty$.
\end{lem}

Let $s_t=s(w_t,R_{t+1})$ be the optimal savings obtained by solving the utility maximization problem \eqref{eq:UMP}. Using Lemma \ref{lem:k}, the market clearing condition \eqref{eq:s_clear}, and Assumption \ref{asmp:G}, we obtain the equilibrium condition
\begin{equation}
    Gk_{t+1}+p_t=s(f(k_t)-k_tf'(k_t),f'(k_{t+1})).
    \label{eq:eqcond}
\end{equation}
\citet{Tirole1985}'s analysis heavily depends on the monotonicity of the dynamical system analogous to \eqref{eq:eqcond}. To this end, \citet{Tirole1985} imposes a high-level assumption on his function $\psi$, which is inherited from \citet{Diamond1965}. (See Appendix \ref{subsec:disc_Tirole}.) Instead, we derive his assumption from explicit restrictions on preferences.

\begin{asmp}\label{asmp:s}
Given $w,R>0$, there exists a unique $s=s(w,R)\in (0,w)$ that maximizes $U(w-s,Rs)$. Furthermore, $s$ is strictly increasing in $w$ and increasing in $R$.
\end{asmp}

Assumption \ref{asmp:s} is still a high-level assumption. The following lemma provides a sufficient condition for Assumption \ref{asmp:s} to hold.

\begin{lem}\label{lem:s}
Suppose that $U(c^y,c^o)=u(c^y)+v(c^o)$, where $u$ is continuously differentiable on $(0,\infty)$, $u'>0$, $u'(0)=\infty$, $u'$ is strictly decreasing, and the same conditions hold for $v$. If $c\mapsto cv'(c)$ is increasing,\footnote{If $v$ is twice differentiable, because $(cv'(c))'=v'(c)+cv''(c)$, it follows that $cv'(c)$ is increasing if and only if $v$ has relative risk aversion (RRA) bounded above by 1, which is a standard condition in general equilibrium theory to establish equilibrium uniqueness \citep{TodaWalsh2024JME}. This parameter restriction is really about the elasticity of intertemporal substitution (EIS), not about RRA, to ensure that the supply of savings is monotonic in the interest rate. See \citet[Remark 1]{GalorRyder1989} and \citet{FlynnSchmidtToda2023TE} for more discussion.} then Assumption \ref{asmp:s} holds.
\end{lem}

Assumption \ref{asmp:s} is substantive. The monotonicity of saving in the
interest rate is used in Lemma \ref{lem:g} to obtain a monotone law of
motion and hence underlies the ordering of equilibria in Proposition \ref{prop:p0} and the global results in Theorems \ref{thm:bubbleless}--\ref{thm:necessity} and \ref{thm:continuum}.
Without such monotonicity, \citet{Tirole1985}'s model may exhibit periodic equilibria; see \citet{Jullien1988}. By contrast, the capital collapse result in Theorem \ref{thm:curse} does not require Assumption \ref{asmp:s}. The following lemma shows that, under the maintained assumptions, the equilibrium condition \eqref{eq:eqcond} can be uniquely solved for $k_{t+1}$, which is monotonic in $k_t$ and $p_t$.\footnote{Lemma \ref{lem:g} is similar to \citet[Proposition 1.3]{delaCroixMichel2002} and \citet*[Claim 1]{BosiHa-HuyLeVanPhamPham2018}.}

\begin{lem}\label{lem:g}
Let $k>0$ and $p\ge 0$. If Assumptions \ref{asmp:G}--\ref{asmp:s} hold, the equation
\begin{equation}
    Gx+p-s(f(k)-kf'(k),f'(x))=0 \label{eq:xeq}
\end{equation}
has at most one solution $x=g(k,p)>0$. Furthermore, letting $\dom g$ be the domain of $g$, the following statements hold.
\begin{enumerate}
    \item\label{item:g1} $(k,0)\in \dom g$ for all $k>0$ and $g(k,0)<k$ for large enough $k>0$.
    \item\label{item:g2} $g$ is continuous, strictly increasing in $k$, and strictly decreasing in $p$.
    \item\label{item:g3} If $(k,p)\in \dom g$, $k'\ge k$, and $0\le p'\le p$, then $(k',p')\in \dom g$ and
    \begin{equation} 
        g(k,p)\le g(k',p'). \label{eq:gineq}
    \end{equation}
\end{enumerate}
\end{lem}

Applying Lemma \ref{lem:g}, we obtain the following proposition, which describes the equilibrium system.

\begin{prop}[Equilibrium system]\label{prop:system}
If Assumptions \ref{asmp:G}--\ref{asmp:s} hold, the equilibrium has a one-to-one correspondence with the system
\begin{subequations}\label{eq:system}
\begin{align}
    k_{t+1}&=g(k_t,p_t), \label{eq:system_k}\\
    p_{t+1}&=\frac{f'(k_{t+1})}{G}p_t-d_{t+1}, \label{eq:system_p}
\end{align}
\end{subequations}
where $k_0>0$ is given, $k_t>0$, $p_t\ge 0$, and $g$ in \eqref{eq:system_k} is defined in Lemma \ref{lem:g}.
\end{prop}

\begin{proof}
The equilibrium condition \eqref{eq:eqcond} and Lemma \ref{lem:g} imply \eqref{eq:system_k}. The no-arbitrage condition \eqref{eq:R} and the definitions of $p_t,d_t$ imply \eqref{eq:system_p}.
\end{proof}

We introduce some terminology. In any equilibrium, we may decompose the asset price as $P_t=V_t+B_t$, where $V_t$ is the fundamental value \eqref{eq:Vt1} and $B_t=P_t-V_t\ge 0$ is the bubble. (See \citet[\S2]{HiranoToda2024JME} for more details.) Accordingly, we may decompose the detrended asset price as $p_t=v_t+b_t\eqqcolon V_t/G^t+B_t/G^t$, where $v_t,b_t$ are the fundamental and bubble components. 
We say that an equilibrium is \emph{bubbly} (\emph{bubbleless}) if $b_t>0$ ($b_t=0$), and \emph{asymptotically bubbly (bubbleless)} if $\liminf_{t\to\infty}b_t>0$ ($=0$).\footnote{\citet{Tirole1985} defined ``asymptotically bubbly'' if the bubble per capita does not converge to zero, meaning $\limsup_{t\to\infty}b_t>0$. He did not clearly define ``asymptotically bubbleless''. Our definition here follows \citet{HiranoToda2025JPE}.} The following Bubble Characterization Lemma, due to \citet[Proposition 7]{Montrucchio2004}, is extremely useful for determining the presence or absence of bubbles.

\begin{lem}[Bubble Characterization]\label{lem:bubble}
If $P_t>0$ for all $t$, then the asset price exhibits a bubble ($P_t>V_t$ for all $t$) if and only if $\sum_{t=1}^\infty D_t/P_t<\infty$.
\end{lem}

Since $k_0>0$ is given, by Proposition \ref{prop:system}, an equilibrium has a one-to-one correspondence with the initial asset price $p_0\ge 0$. For this reason, we say that $p_0\ge 0$ is an equilibrium if it corresponds to an equilibrium system in Proposition \ref{prop:system} and denote the equilibrium set by $\cP_0$. Obviously, Proposition \ref{prop:exist} implies $\cP_0\neq\emptyset$. We sometimes say $p_0$ is \emph{bubbly} (\emph{bubbleless}) if the equilibrium corresponding to $p_0$ is bubbly (bubbleless).

The following proposition establishes monotonicity properties of the equilibrium, which play an important role in the subsequent analysis.

\begin{prop}[Equilibrium monotonicity]\label{prop:p0}
Suppose Assumptions \ref{asmp:G}--\ref{asmp:s} hold and let $\cP_0$ be the equilibrium set. Then the following statements hold.
\begin{enumerate}
    \item\label{item:p0_interval} $\cP_0$ is a nonempty compact interval.
    \item\label{item:p0_monotone} Let  $p_0,p_0'\in \cP_0$ and $p_0<p_0'$. Let $\set{(k_t,p_t,R_t,w_t)}_{t=0}^\infty$ satisfy the equilibrium system \eqref{eq:system} and let $p_t=v_t+b_t$ be the fundamental-bubble decomposition. Define $(k_t',p_t',R_t',w_t',v_t',b_t')$ analogously. Then for all $t\ge 1$ we have
    \begin{align*}
        k_t&>k_t', & p_t&<p_t', & R_t&<R_t',\\
        w_t&>w_t', & v_t&\ge v_t', & b_t&<b_t'.
    \end{align*}
\end{enumerate}
\end{prop}

\begin{cor}[Uniqueness of bubbleless equilibrium]\label{cor:unique_bubbleless}
There exists at most one bubbleless equilibrium, which corresponds to $p_0=\min \cP_0$.
\end{cor}

\begin{proof}
If $p_0,p_0'\in \cP_0$ and $p_0<p_0'$, by Proposition \ref{prop:p0}\ref{item:p0_monotone} we have $b_t'>b_t\ge 0$, so $p_0'$ is bubbly. Therefore, if a bubbleless equilibrium exists, it must be $p_0=\min \cP_0$.
\end{proof}

\section{Global characterization of equilibria}\label{sec:main}

This section characterizes the long-run behavior of equilibrium in various settings. It is instructive to start with the special case of the \citet{Diamond1965} model without the asset, which satisfies the one-dimensional dynamics $k_{t+1}=g(k_t,0)$ by Proposition \ref{prop:system}. The following lemma shows that capital always converges in the \citet{Diamond1965} model, which bounds the capital sequence in any equilibrium in the \citet{Tirole1985} model.

\begin{lem}\label{lem:diamond}
Suppose Assumptions \ref{asmp:G}--\ref{asmp:s} hold and define the set of steady-state capital-labor ratios without the asset by
\begin{equation}
    \cK^*\coloneqq \set{k>0:k=g(k,0)}. \label{eq:cK}
\end{equation}
For any $k_0>0$, define the sequence $\set{k_t^*}_{t=0}^\infty\subset (0,\infty)$ by $k_0^*=k_0$ and $k_{t+1}^*=g(k_t^*,0)$. Then:
\begin{enumerate}
    \item\label{item:diamond_converge} $k_t^*\to k^*\in \set{0}\cup \cK^*$. If $g(k,0)>k$ for sufficiently small $k>0$, then $\cK^*$ is compact and $k^*\in \cK^*$.
    \item\label{item:diamond_ub} In any equilibrium $\set{(k_t,p_t)}_{t=0}^\infty$, we have $k_t\le k_t^*$ for all $t$.
\end{enumerate}
\end{lem}

\subsection{Existence of bubbleless equilibria}

The following theorem shows that a unique bubbleless equilibrium exists when interest rates are sufficiently high or dividends are sufficiently small. In what follows, it is convenient to define the long-run dividend growth rate
\begin{equation}
    G_d\coloneqq \limsup_{t\to\infty} D_t^{1/t}. \label{eq:Gd}
\end{equation}

\begin{thm}[Existence of bubbleless equilibrium]\label{thm:bubbleless}
Suppose Assumptions \ref{asmp:G}--\ref{asmp:s} hold. For $k_0>0$, let $\set{k_t^*}_{t=0}^\infty$ be as in Lemma \ref{lem:diamond} and $R_t^*=f'(k_t^*)$. If
\begin{equation}
    V_0^*\coloneqq \sum_{t=1}^\infty \frac{D_t}{R_1^*\dotsb R_t^*}<\infty, \label{eq:bubbleless_cond}
\end{equation}
then there exists a unique bubbleless equilibrium. In particular, if $G_d$ in \eqref{eq:Gd} satisfies $G_d<R^*\coloneqq f'(k^*)$, then \eqref{eq:bubbleless_cond} holds.
\end{thm}

\citet[Lemma 1]{Tirole1985}  claims that under the condition $R^*>G_d$ in our notation, there exists a unique bubbleless equilibrium and that $\lim_{t\to\infty}R_t=R^*$. We find this lemma problematic. (See Appendix \ref{subsec:disc_Tirole} for details.) Theorem \ref{thm:bubbleless} is more general than Lemma 1 in \citet{Tirole1985} and our proof is new.

Recall that we denote detrended dividend by $d_t=D_t/G^t$. The following two lemmas show that when dividends or interest rates are sufficiently large, bubbles are impossible.

\begin{lem}[Impossibility of bubbles, I]\label{lem:impossible1}
If Assumptions \ref{asmp:G}, \ref{asmp:f} hold and $\sum_{t=1}^\infty d_t=\infty$, then all equilibria are bubbleless.
\end{lem}

\begin{proof}
By Lemma \ref{lem:k}, we can take $p>0$ such that $p_t\le p$ for all $t$. Since $\sum_{t=1}^\infty d_t=\infty$, we have $d_t>0$ infinitely often, so $p_t>0$ for all $t$. By assumption,
\begin{equation*}
    \sum_{t=1}^\infty \frac{D_t}{P_t}=\sum_{t=1}^\infty \frac{d_t}{p_t}\ge \sum_{t=1}^\infty \frac{d_t}{p}=\infty,
\end{equation*}
so the equilibrium is bubbleless by Lemma \ref{lem:bubble}.
\end{proof}

Lemma \ref{lem:impossible1} is consistent with Theorem 3.3 of \citet{SantosWoodford1997}, which precludes bubbles when dividends are non-negligible relative to aggregate resources.

\begin{lem}[Impossibility of bubbles, II; Lemma 2.3 of \citealp{PhamToda2026ECMA}]\label{lem:impossible2}
Suppose Assumptions \ref{asmp:G}--\ref{asmp:s} hold and let $\set{(k_t,p_t)}_{t=0}^\infty$ be an equilibrium. If $\bar{k}=\limsup_{t\to\infty}k_t$ and $f'(\bar{k})>G$, then the equilibrium is unique and bubbleless.
\end{lem}

Theorem \ref{thm:bubbleless} states that a unique bubbleless equilibrium exists if dividend growth is sufficiently low. Lemmas \ref{lem:impossible1} and \ref{lem:impossible2} (which correspond to Corollary 2 and Proposition 2.1 in \citealp{BosiHa-HuyLeVanPhamPham2018}) show that equilibria are necessarily bubbleless if either the dividend growth or the bubbleless interest rate is sufficiently high.

\subsection{Bubble possibility and necessity}

We next study under what conditions bubbles are possible (``bubbles can arise'') or necessary (``bubbles must arise''). Lemma \ref{lem:impossible1} implies that $\sum_{t=1}^\infty d_t<\infty$ is necessary for the existence of bubbles. We thus assume this condition.

\begin{asmp}\label{asmp:D}
$\sum_{t=1}^\infty d_t<\infty$.
\end{asmp}

Assumption \ref{asmp:D} implies that in the long run, dividends become negligible relative to the economy. If we define the long-run dividend growth rate by $G_d$ in \eqref{eq:Gd}, then Assumption \ref{asmp:D} implies $G_d\le G$. The following proposition characterizes the possible long-run behavior of the equilibrium.

\begin{prop}[Long-run behavior of equilibrium]\label{prop:longrun}
Suppose Assumptions \ref{asmp:G}--\ref{asmp:D} hold. In any equilibrium, one of the following statements is true.
\begin{enumerate}
    \item\label{item:lr_bubbleless} The equilibrium is bubbleless and $\set{(k_t,p_t,R_t)}$ converges to $(k,0,R)$ satisfying $k\in \set{0}\cup \cK^*$ and $R=f'(k)\ge G$.
    \item\label{item:lr_asymbubbleless} The equilibrium is asymptotically bubbleless and $\set{(k_t,p_t,R_t)}$ converges to $(k,0,R)$ satisfying $k\in \cK^*$ and $R=f'(k)\in [G_d,G]$.
    \item\label{item:lr_bubbly} The equilibrium is asymptotically bubbly and $\set{(k_t,p_t,R_t)}$ converges to $(k,p,G)$ satisfying $k=g(k,p)$, $p>0$, and $G=f'(k)$.
\end{enumerate}
\end{prop}

Note that statements \ref{item:lr_bubbleless} and \ref{item:lr_asymbubbleless} in Proposition \ref{prop:longrun} are not mutually exclusive because there could be a bubbleless equilibrium with $R_t\to G$. Statement \ref{item:lr_bubbly} completely characterizes the long-run behavior of the asymptotically bubbly equilibrium. As a corollary, we obtain its uniqueness.

\begin{cor}[Uniqueness of asymptotically bubbly equilibrium]\label{cor:unique_bubble}If Assumptions \ref{asmp:G}--\ref{asmp:D} hold, there exists at most one asymptotically bubbly equilibrium, which corresponds to $p_0=\max \cP_0$.
\end{cor}

\begin{proof}
Uniqueness follows from Lemma A.3 of \citet{PhamToda2026ECMA}. If $p_0\in \cP_0$ is asymptotically bubbly, Proposition \ref{prop:p0}\ref{item:p0_monotone} forces $p_0=\max \cP_0$.
\end{proof}

Proposition \ref{prop:longrun} describes the long-run behavior of a particular equilibrium. We are now interested in characterizing all possible forms of the equilibrium set $\cP_0$. The following theorem provides all possible forms of the equilibrium set $\cP_0$ when the steady-state interest rate is low (capital over-accumulation).

\begin{thm}[Trichotomy under capital over-accumulation]\label{thm:eqset}
Suppose Assumptions \ref{asmp:G}--\ref{asmp:D} hold and $g(k,0)>k$ for small enough $k>0$. If
\begin{equation}
    f'(k)<G \label{eq:over-accumulation}
\end{equation}
for all $k\in \cK^*$, then exactly one of the following statements is true.
\begin{enumerate}
    \item\label{item:eqset1} (Capital collapse) There exists a unique equilibrium, which is bubbleless and $\set{(k_t,p_t,R_t)}$ converges to $(0,0,\infty)$.

    \item\label{item:eqset2} (Bubble necessity) There exists a unique equilibrium, which is asymptotically bubbly and $\set{(k_t,p_t,R_t)}$ converges to $(k,p,G)$ satisfying $k=g(k,p)$, $p>0$, and $G=f'(k)$.

    \item\label{item:eqset3} (Indeterminacy) There exists a continuum of equilibria and the equilibrium set is a compact interval $\cP_0=[\ubar{p}_0,\bar{p}_0]$.
    \begin{enumerate}
        \item\label{item:eqset3a} If $p_0>\ubar{p}_0$, then the equilibrium is bubbly.
        \item\label{item:eqset3b} If $p_0\in [\ubar{p}_0,\bar{p}_0)$, then the equilibrium is asymptotically bubbleless and $\set{(k_t,p_t,R_t)}$ converges to $(k,0,R)$ satisfying $k=g(k,0)>0$ and $R=f'(k)$.
        \item\label{item:eqset3c} If $p_0=\bar{p}_0$, then the equilibrium is asymptotically bubbly and $\set{(k_t,p_t,R_t)}$ converges to $(k,p,G)$ satisfying $k=g(k,p)$, $p>0$, and $G=f'(k)$.
    \end{enumerate}
\end{enumerate}
\end{thm}

\begin{rem}
To appreciate the significance of Theorem \ref{thm:eqset}, it is useful to recall the situation without the asset. In the \citet{Diamond1965} model, the dynamics are one-dimensional and capital converges by Lemma \ref{lem:diamond}. With the asset, the problem becomes two-dimensional: given initial capital $k_0$, the initial
price $p_0$ is a free variable, and different choices of $p_0\in \cP_0$ lead to different long-run
outcomes. Theorem \ref{thm:eqset} completely characterizes this. Importantly, \citet{Tirole1985} did not consider the possibility of case \ref{item:eqset1}; as we shall see in Theorem \ref{thm:curse} below, this case arises robustly. Thus, the introduction of the dividend-paying asset fundamentally changes the economic intuition.
\end{rem}

\begin{rem}
In case \ref{item:eqset3}, the equilibrium conditions characterize but do not select the initial asset price: each $p_0\in \cP_0$ is supported by a
corresponding path of price expectations. Indeterminacy of this kind is a well-known feature of OLG models \citep{Gale1973}. Theorem \ref{thm:necessity} below identifies conditions under which the indeterminacy disappears, whereas Theorem \ref{thm:continuum} identifies conditions under which it must occur. Note that in case \ref{item:eqset3b}, we do not claim $p_0=\ubar{p}_0$ is bubbleless, and indeed, it could be bubbly (Appendix \ref{sec:bubbly-minimum}).
\end{rem}

We provide a few remarks regarding the additional assumptions in Theorem \ref{thm:eqset}.

\begin{rem}\label{rem:singlecross}
If we map \citet{Tirole1985}'s notation to ours (see Appendix \ref{subsec:disc_Tirole}), he assumes that the set $\cK^*=\set{k^*}$ in \eqref{eq:cK} is a singleton and that $k\gtrless g(k,0)$ according as $k\gtrless k^*$. This single-crossing condition is relatively strong and it is not easy to provide sufficient conditions based only on exogenous objects. It holds with constant-elasticity-of-substitution (CES) production functions with elasticity at least 1, but not with CES production functions with elasticity less than 1 \citep{delaCroixMichel2002, HiranoToda2024EL}. In contrast, we do not impose any condition on $\cK^*$.
\end{rem}

\begin{rem}\label{rem:kubar}
The condition $g(k,0)>k$ for small enough $k>0$ in Theorem \ref{thm:eqset} is standard in the literature \citep[pp.~34--36]{delaCroixMichel2002}.\footnote{\label{fn:g}\citet[p.~36, Proposition 1.7]{delaCroixMichel2002} provide a sufficient condition. In our context, it suffices to assume $U$ is additively separable as $U(c^y,c^o)=(1-\beta)u(c^y)+\beta u(c^o)$ and $\lim_{k\to 0}\omega(k)/k>G/\beta$, where $\omega(k)\coloneqq f(k)-kf'(k)$ is the wage function. See also \citet{GalorRyder1989}, who show that a condition of this type is necessary to preclude convergence to $k=0$.} If $g(k,0)<k$ for all $k\in (0,\bar{k})$, then for $k_0<\bar{k}$, Lemma \ref{lem:diamond}\ref{item:diamond_converge} implies $k_t^*\to 0$, and Lemma \ref{lem:diamond}\ref{item:diamond_ub} implies $0\le k_t\le k_t^*\to 0$ in every equilibrium. Hence Lemma \ref{lem:impossible2} gives uniqueness and only case \ref{item:eqset1} of Theorem \ref{thm:eqset} occurs.
\end{rem}

Theorem \ref{thm:necessity} below states that when the dividend growth rate is intermediate, asymptotically bubbleless equilibria (case \ref{item:eqset3} in Theorem \ref{thm:eqset}) are ruled out, and the equilibrium is unique and either bubbleless or asymptotically bubbly. Note that for this result, we do not require $g(k,0)>k$ for small enough $k>0$.

\begin{thm}[Collapse-bubble dichotomy under bubble necessity condition]\label{thm:necessity}
Suppose Assumptions \ref{asmp:G}--\ref{asmp:D} hold and let $k^*\coloneqq \lim_{t\to\infty}k_t^*$ in Lemma \ref{lem:diamond}. If
\begin{equation}
    f'(k)<G_d \label{eq:necessity}
\end{equation}
for all $k\in \cK^*\cap (0,k^*]$, then the equilibrium is unique and either case \ref{item:eqset1} or \ref{item:eqset2} of Theorem \ref{thm:eqset} occurs. If in addition \eqref{eq:necessity} holds for all $k\in \cK^*$, then there exists $\kappa\in [0,\infty]$ such that case \ref{item:eqset1} occurs if $k_0<\kappa$ and case \ref{item:eqset2} occurs if $k_0>\kappa$.
\end{thm}

Theorem \ref{thm:eqset} provides all possible forms of the equilibrium set. Theorem \ref{thm:necessity} shows that the bubble necessity condition $R<G_d<G$ in \citet{HiranoToda2025JPE} rules out case \ref{item:eqset3}. Furthermore, whether case \ref{item:eqset1} or \ref{item:eqset2} occurs depends on whether initial capital is small or large. (No claim is made when $k_0=\kappa$.) Theorem 1 of \citet{PhamToda2026ECMA} shows that $\kappa<\infty$ if $\sup_t d_t$ is small enough, so case \ref{item:eqset2} can occur.

The following theorem shows that under some explicit conditions on the production function near $k=0$ and dividend growth, capital converges to zero regardless of preferences if initial capital is sufficiently small. In particular, $\kappa>0$ in Theorem \ref{thm:necessity} and case \ref{item:eqset1} can occur.

\begin{thm}[Capital collapse]\label{thm:curse}
Suppose Assumptions \ref{asmp:G} and \ref{asmp:f} hold and let $\omega(k)\coloneqq f(k)-kf'(k)$ be the wage function. Suppose there exist $0<D_1\le D_2$ and $0<r_1\le r_2<1$ such that $D_1r_1^t\le d_t\le D_2r_2^t$ for all $t$ and
\begin{equation}
    L\coloneqq \lim_{k\to 0}k^{-\alpha}\omega(k)\in (0,\infty) \label{eq:w_lim}
\end{equation}
for some
\begin{equation}
    \frac{1}{1+\log r_2/\log r_1}<\alpha\le 1. \label{eq:alpha}
\end{equation}
Then there exists $\kappa>0$ such that, if $k_0<\kappa$, we have $k_t\to 0$ in any equilibrium.
\end{thm}

To show that there is nothing unusual about the condition \eqref{eq:w_lim}, we present two particular cases.

\begin{exmp}[$\alpha<1$] \label{curse_example1}
Consider the Cobb-Douglas production function (with partial capital depreciation) $f(k)=Ak^\alpha+(1-\delta)k$, where $A>0$, $\alpha\in (0,1)$, and $\delta\in [0,1]$. Then $\omega(k)=f(k)-kf'(k)=A(1-\alpha)k^\alpha$, so \eqref{eq:w_lim} holds.
\end{exmp}

\begin{exmp}[$\alpha=1$] \label{curse_example2}
To construct a counterexample to Proposition 1(c) of \citet{Tirole1985}, \citet{PhamToda2026ECMA} consider $f(k)=Ak\log(1+1/k)$, where $A>0$. Then $\omega(k)=f(k)-kf'(k)=Ak/(1+k)$, so \eqref{eq:w_lim} holds with $\alpha=1$.
\end{exmp}

\begin{rem}\label{rem:collapse}
Theorem \ref{thm:curse} is preference-free and the condition \eqref{eq:w_lim} depends only on the behavior of the production function near $k=0$, so it is quite robust. The conclusion of Theorem \ref{thm:curse} also holds if $\lim_{k\to 0}\omega(k)/k=0$ with no restrictions on dividends.\footnote{Here is the proof. Take $\kappa>0$ and $r\in (0,1)$ such that $\omega(k)/k\le rG$ for $k\in (0,\kappa]$. Suppose $k_0\le \kappa$ and take any equilibrium. By feasibility, $Gk_1\le Gk_1+p_0\le \omega(k_0)$, so $k_1\le rk_0\le \kappa$. By induction, we have $k_t\le r^t k_0$ for all $t$, so $k_t\to 0$.} Under constant aggregate dividends, $d_t=DG^{-t}$, so \eqref{eq:alpha} reduces to $\alpha\in (1/2,1]$. Thus, Theorem \ref{thm:curse} identifies a broad class of technologies---including Cobb-Douglas---for which sufficiently small initial capital leads to capital collapse. As $f$ in Theorem \ref{thm:curse} is essentially unrestricted away from $k=0$ and the possibility of capital collapse is absent in Proposition 1(a)(b)(c) of \citet{Tirole1985}, these results are incorrect as stated. It is worth being precise about what drives this collapse. When $\alpha<1$ and the preference conditions in Footnote \ref{fn:g} hold, we have $\omega(k)/k\to\infty$ as $k\to 0$, so $g(k,0)>k$ near the origin and the Diamond economy converges to a positive steady state (Lemma \ref{lem:diamond}). Thus, for this subclass, the collapse in Theorem \ref{thm:curse} is not inherited from the Diamond dynamics; it is induced by the dividend-paying asset's crowding out of capital investment. The intuition for the capital collapse in the presence of a dividend-paying asset is as follows. The budget constraint of the young implies that wage is spent on consumption, future capital, and asset purchase. If dividends are not too small, the asset price cannot be too small, so future capital must be small. Thus, it may be possible to sustain an equilibrium in which capital collapses.
\end{rem}

Finally, we provide conditions under which case \ref{item:eqset3} in Theorem \ref{thm:eqset} must occur. To this end, it is convenient to define
\begin{equation}
    p(k)\coloneqq s(f(k)-kf'(k),f'(k))-Gk. \label{eq:pk}
\end{equation}
Using \eqref{eq:eqcond}, we can interpret $p(k)$ as the asset price consistent with steady-state capital $k$. Let $\Phi(x,k,p)$ be the left-hand side of \eqref{eq:xeq}. By the proof of Lemma \ref{lem:g}, $\Phi$ is strictly increasing in $x,p$ and strictly decreasing in $k$. Therefore, if $g(k,0)>k$, by the definition of $g$, we have
\begin{equation*}
    p(k)=-\Phi(k,k,0)>-\Phi(g(k,0),k,0)=0.
\end{equation*}
Furthermore, the definition of $g$ and $p(k)$ in \eqref{eq:pk} imply $k=g(k,p(k))$.

The following theorem shows that, if initial capital is sufficiently large and the present discounted value of dividends under the \citet{Diamond1965} interest rate is small enough, then there exists a continuum of equilibria.

\begin{thm}[Existence of a continuum of equilibria]\label{thm:continuum}
Let everything be as in Theorem \ref{thm:eqset}, $k_g>0$ be the golden rule capital defined by $f'(k_g)=G$, and $V_0^*$ be as in \eqref{eq:bubbleless_cond}. Then $p(k_g)>0$. If $k_0\ge k_g$ and $V_0^*\le p(k_g)$, then case \ref{item:eqset3} in Theorem \ref{thm:eqset} occurs and $p_0=\min \cP_0$ corresponds to the unique bubbleless equilibrium.
\end{thm}

Theorem \ref{thm:continuum} implies that the continuum of pure bubble equilibria in the \citet{Tirole1985} model with $D=0$ survives sufficiently small dividend perturbations.

\section{Analytical examples}\label{sec:example}

In this section, we present analytical examples to illustrate the asymptotic behavior in Theorem \ref{thm:eqset}. Consider the Cobb-Douglas production function \eqref{eq:CD} with full capital depreciation ($\delta=1$), so $f(k)=Ak^\alpha$. Obviously, Assumptions \ref{asmp:G}, \ref{asmp:f} hold. Furthermore, assume log utility
\begin{equation}
    U(c^y,c^o)=(1-\beta)\log c^y+\beta \log c^o, \label{eq:utility_log}
\end{equation}
where $\beta\in (0,1)$. Then the savings function in Lemma \ref{lem:s} admits the closed-form expression $s(w,R)=\beta w$ and Assumption \ref{asmp:s} holds.

Under this specification, the equilibrium system \eqref{eq:system} reduces to
\begin{subequations}\label{systemCD}
\begin{align}
    k_{t+1}&=\frac{\beta A(1-\alpha)k_t^\alpha-p_t}{G}, \label{eq:systemCD_k}\\
    p_{t+1}&=\frac{A\alpha k_{t+1}^{\alpha-1}}{G}p_t-d_{t+1}. \label{eq:systemCD_p}
\end{align}
\end{subequations}
We obtain the bubbleless steady state by solving $k=g(k,0)$ for $k>0$, or
\begin{equation}
    k=\frac{\beta A(1-\alpha)}{G}k^\alpha\iff k^*=\left(\frac{\beta A(1-\alpha)}{G}\right)^\frac{1}{1-\alpha}. \label{eq:k_bubbleless}
\end{equation}
Note that the set $\cK^*$ in \eqref{eq:cK} is a singleton consisting of $k^*$ and $g(k,0)>k$ for $k\in (0,k^*)$. The steady-state interest rate is
\begin{equation*}
    R=f'(k^*)=A\alpha (k^*)^{\alpha-1}=\frac{\alpha G}{\beta(1-\alpha)}.
\end{equation*}
To see whether condition \eqref{eq:over-accumulation} in Theorem \ref{thm:eqset} is satisfied, we define
\begin{equation}
    \rho\coloneqq \frac{R}{G}=\frac{\alpha}{\beta(1-\alpha)}. \label{eq:rho}
\end{equation}

The following lemma, which is essentially a change of variables and similar to Step 1 in the proof of Example 1 in \citet{BosiHa-HuyLeVanPhamPham2018}, allows us to construct many examples.

\begin{lem}\label{lem:CD}
Assume log utility \eqref{eq:utility_log} and Cobb-Douglas production function \eqref{eq:CD} with $\delta=1$. Let $k_0>0$ be given and take any $d_0\ge 0$. Then the following statements hold.
\begin{enumerate}
    \item\label{item:lem_CD1} For any equilibrium $\set{(k_t,p_t)}_{t=0}^\infty$ with $p_t>0$, $x_t\coloneqq A\alpha k_t^\alpha/(Gk_{t+1})$ satisfies
    \begin{subequations}\label{eq:xineq}
    \begin{align}
        x_t&>\rho, \label{eq:xineq1}\\
        x_t+\frac{\rho}{x_{t+1}}&\ge 1+\rho. \label{eq:xineq2}
    \end{align}
    \end{subequations}
    \item\label{item:lem_CD2} For any sequence $\set{x_t}_{t=0}^\infty$ satisfying \eqref{eq:xineq}, if we define
    \begin{subequations}\label{eq:xdef}
    \begin{align}
        k_{t+1}&=\frac{A\alpha k_t^\alpha}{Gx_t}, \label{eq:xdef_k}\\
        p_t&=\frac{A\alpha}{\rho}k_t^\alpha-Gk_{t+1}, \label{eq:xdef_p}\\
        d_{t+1}&=\frac{A\alpha}{G}k_{t+1}^{\alpha-1}p_t-p_{t+1}, \label{eq:xdef_d}
    \end{align}
    \end{subequations}
    the sequence $\set{(k_t,p_t)}_{t=0}^\infty$ is an equilibrium for the model with dividend $D_t=d_tG^t$ given by \eqref{eq:xdef_d}. Furthermore, we have
    \begin{subequations}\label{eq:kdx}
    \begin{align}
        k_{t+1}&=\left(\frac{A\alpha}{G}\right)^\frac{1-\alpha^{t+1}}{1-\alpha}k_0^{\alpha^{t+1}}\frac{1}{x_tx_{t-1}^\alpha\dotsb x_0^{\alpha^t}}, \label{eq:kdx_k}\\
        d_{t+1}&=\frac{A\alpha}{\rho}k_{t+1}^\alpha\left(x_t+\frac{\rho}{x_{t+1}}-1-\rho\right),\label{eq:kdx_d} \\
        \frac{d_t}{p_t}&=\frac{x_{t-1}+\rho/x_t-1-\rho}{1-\rho/x_t}. \label{eq:kdx_dp} 
    \end{align}
    \end{subequations}
\end{enumerate}
\end{lem}

Furthermore, the following lemma is useful.

\begin{lem}\label{lem:kx}
Let $\set{x_t}_{t=0}^\infty$ be a positive sequence converging to $x\in (0,\infty]$. For any $k_0>0$, define $\set{k_t}_{t=1}^\infty$ by \eqref{eq:xdef_k}. Then $k_t\to k=\left(\frac{A\alpha}{Gx}\right)^\frac{1}{1-\alpha}$.
\end{lem}

\subsection{Capital collapse}

We first construct an equilibrium with $k_t\to 0$. By Lemma \ref{lem:kx}, we need $x_t\to\infty$. Thus, set $x_t=C\sigma^t$, where $C>0$ and $\sigma>1$. Set $C\ge 1+\rho$ so that \eqref{eq:xineq} holds. Therefore, we may construct an equilibrium using Lemma \ref{lem:CD}\ref{item:lem_CD2}. Iterating \eqref{eq:kdx}, we obtain $\log d_t=\mu_t\log C+\nu_t\log\sigma+O(1)$, where
\begin{align*}
    \mu_t&=1-(\alpha+\alpha^2+\dots+\alpha^t)=1-\alpha\frac{1-\alpha^t}{1-\alpha},\\
    \nu_t&=(t-1)-(\alpha(t-1)+\alpha^2(t-2)+\dots+\alpha^t\cdot 0)\\
    &=\frac{1-2\alpha}{1-\alpha}(t-1)+\left(\frac{\alpha}{1-\alpha}\right)^2(1-\alpha^{t-1}).
\end{align*}
It follows that
\begin{equation}
    \lim_{t\to\infty}d_t^{1/t}=\lim_{t\to\infty}C^{\mu_t/t}\sigma^{\nu_t/t}=\sigma^\frac{1-2\alpha}{1-\alpha}. \label{eq:d3}
\end{equation}
Since $\sigma>1$, Assumption \ref{asmp:D} holds by setting $\alpha>1/2$. Therefore, we obtain the following result.

\begin{exmp}[Unique, bubbleless equilibrium with $k_t\to 0$]\label{exmp:k0}
Let everything be as in Lemma \ref{lem:CD}. Let $\alpha>1/2$ and define $\rho>1/\beta$ by \eqref{eq:rho}. For any $C\ge 1+\rho$ and $\sigma>1$, let $x_t=C\sigma^t$ and define $\set{(k_{t+1},p_t,d_{t+1})}_{t=0}^\infty$ by \eqref{eq:xdef}. Then $\set{(k_t,p_t)}_{t=0}^\infty$ is the unique equilibrium of the economy with detrended dividend $\set{d_t}_{t=0}^{\infty}$, which converges to $(0,0)$.
\end{exmp}

Note that $k_t\to 0$ follows from $x_t\to\infty$ and Lemma \ref{lem:kx}; $p_t\to 0$ follows from \eqref{eq:xdef_p}; $d_t\to 0$ follows from \eqref{eq:d3} and $\alpha>1/2$; and the uniqueness of equilibrium follows from Lemma \ref{lem:impossible2}. Figure \ref{fig:bubbleless} shows the time path of $\set{(k_t,p_t,d_t)}$ for a numerical example. We set $G=1$ and $A=1/(\beta(1-\alpha))$ so that there is no growth and the steady-state capital is normalized to 1.

\begin{figure}[!htb]
\centering
\includegraphics[width=0.7\linewidth]{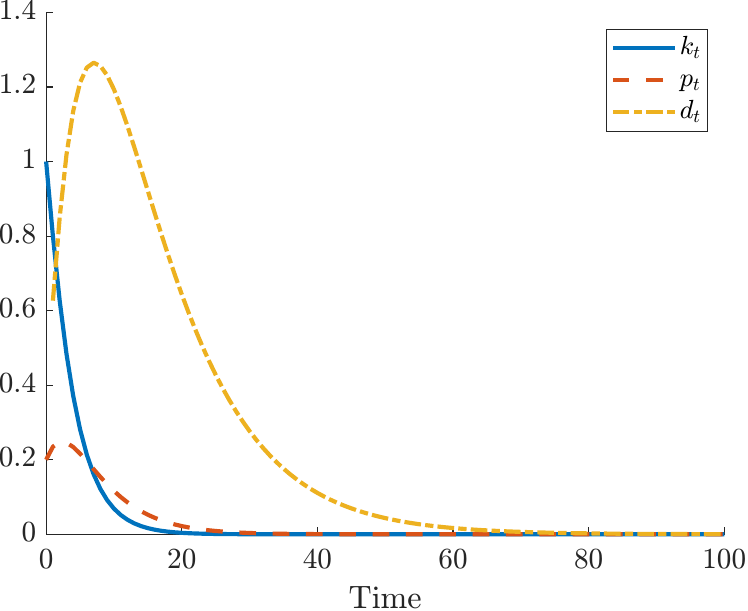}
\caption{Bubbleless equilibrium with $k_t\to 0$.}\label{fig:bubbleless}
\caption*{\footnotesize Note: Parameter values in Example \ref{exmp:k0} are $G=1$, $\alpha=2/3$, $\beta=1/2$, $A=1/(\beta(1-\alpha))$, $C=1+\rho$, $\sigma=1.1$, and $k_0=1$.}
\end{figure}

Although the conclusion of Example \ref{exmp:k0} is the same as Theorem \ref{thm:eqset}\ref{item:eqset1}, it does not satisfy its hypotheses. This is because $\alpha>1/2$ and $\beta\in (0,1)$ force $\rho=\frac{\alpha}{\beta(1-\alpha)}>1$, so $R>G$ and the condition \eqref{eq:over-accumulation} fails. Nevertheless, case \ref{item:eqset1} in Theorem \ref{thm:eqset} is possible, as Proposition 1 of \citet{PhamToda2026ECMA} shows.
\subsection{Bubble necessity}

We next seek an example of Theorem \ref{thm:eqset}\ref{item:eqset2}. If such an equilibrium exists, we have $G=f'(k)=A\alpha k^{\alpha-1}$. Using \eqref{eq:xdef_k}, it must be $x_t\to 1$. Therefore, set $x_t=1+C\sigma^t$ for some constants $C>0$ and $\sigma\in (0,1)$. For condition \eqref{eq:necessity} to hold, set $\alpha<\frac{\beta}{1+\beta}$ so that $\rho$ in \eqref{eq:rho} satisfies $\rho\in (0,1)$. Then \eqref{eq:xineq1} clearly holds. To check \eqref{eq:xineq2}, we compute
\begin{align}
    x_t+\frac{\rho}{x_{t+1}}-1-\rho&=1+C\sigma^t+\frac{\rho}{1+C\sigma^{t+1}}-1-\rho \notag \\
   &=C\sigma^t\left(1-\frac{\rho\sigma}{1+C\sigma^{t+1}}\right). \label{eq:xineq2diff}
\end{align}
Because $C>0$ and $\rho,\sigma\in (0,1)$, \eqref{eq:xineq2} holds. Therefore, we may construct an equilibrium using Lemma \ref{lem:CD}\ref{item:lem_CD2}. This equilibrium is bubbly. To see why, note that the denominator of \eqref{eq:kdx_dp} converges to $1-\rho>0$ as $t\to\infty$. Using \eqref{eq:xineq2diff}, the numerator of \eqref{eq:kdx_dp} has the order of magnitude $C(1-\rho\sigma)\sigma^{t-1}$ as $t\to\infty$. Therefore, $\sum_{t=1}^\infty d_t/p_t<\infty$, so by Lemma \ref{lem:bubble} the equilibrium is bubbly. To apply Theorem \ref{thm:necessity} (whose conclusions are stronger than Theorem \ref{thm:eqset}), it remains to verify condition \eqref{eq:necessity}.

Since $x_t\to 1$, by Lemma \ref{lem:kx}, we have $k_t\to k=(A\alpha/G)^\frac{1}{1-\alpha}$, so case \ref{item:eqset1} in Theorem \ref{thm:eqset} is ruled out. By \eqref{eq:kdx_d} and \eqref{eq:xineq2diff}, $d_t/\sigma^t$ converges to a positive constant. Therefore,
\begin{equation*}
    \frac{G_d}{G}=\limsup_{t\to\infty}(D_t/G^t)^{1/t}=\limsup_{t\to\infty}d_t^{1/t}=\sigma
\end{equation*}
and Assumption \ref{asmp:D} holds. Since by definition $\rho=R/G$, condition \eqref{eq:necessity} holds if and only if $\rho<\sigma<1$. Therefore, we obtain the following result.

\begin{exmp}[Unique, asymptotically bubbly equilibrium]\label{exmp:bubbly}
Let everything be as in Lemma \ref{lem:CD}. Let $\alpha<\frac{\beta}{1+\beta}$ so that $\rho$ in \eqref{eq:rho} satisfies $\rho\in (0,1)$. For any $C>0$ and $\sigma\in (\rho,1)$, let $x_t=1+C\sigma^t$ and define $\set{(k_{t+1},p_t,d_{t+1})}_{t=0}^\infty$ by \eqref{eq:xdef}. Then $\set{(k_t,p_t)}_{t=0}^\infty$ is the unique equilibrium of the economy and the conclusion of Theorem \ref{thm:eqset}\ref{item:eqset2} holds.
\end{exmp}

Example \ref{exmp:bubbly} shows that case \ref{item:eqset2} in Theorem \ref{thm:eqset} is possible. As far as we are aware, Example \ref{exmp:bubbly} is the first closed-form example in which \citet{Tirole1985}'s model has a unique asymptotically bubbly equilibrium. Figure \ref{fig:bubbly} shows the time path of $\set{(k_t,p_t,d_t)}$ for a numerical example. We set $G=1$ and $A=1/(\beta(1-\alpha))$ so that there is no growth and the steady-state capital is normalized to 1. We start the economy at the steady state $k_0=1$. As dividends are initially high, capital undershoots and then converges to the bubbly steady-state value, while the asset price converges to a positive value although the detrended dividend converges to zero.

\begin{figure}[!htb]
\centering
\includegraphics[width=0.7\linewidth]{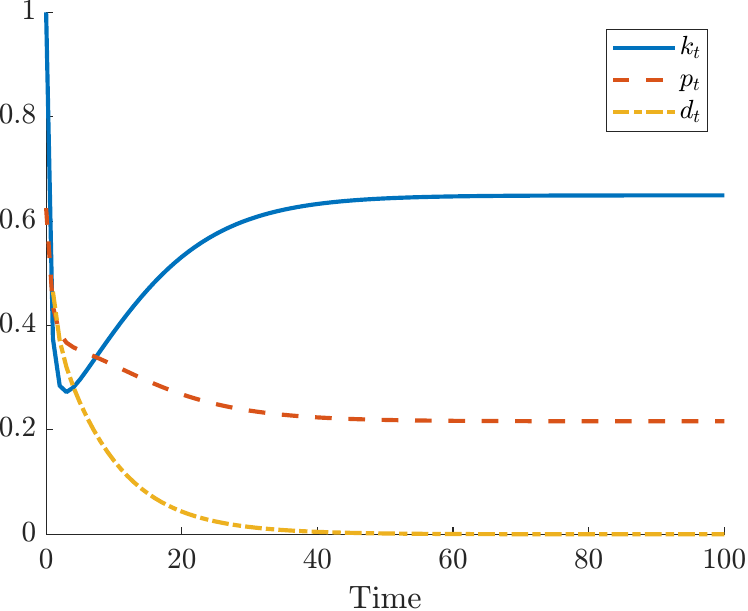}
\caption{Asymptotically bubbly equilibrium.}\label{fig:bubbly}
\caption*{\footnotesize Note: Parameter values in Example \ref{exmp:bubbly} are $G=1$, $\alpha=1/3$, $\beta=2/3$, $A=1/(\beta(1-\alpha))$, $C=1$, $\sigma=0.9$, and $k_0=1$.}
\end{figure}

\subsection{Indeterminacy}

Finally, we seek an example of Theorem \ref{thm:eqset}\ref{item:eqset3}. If a bubbly but asymptotically bubbleless equilibrium exists, we have $k_t\to k^*$ given by \eqref{eq:k_bubbleless}. By \eqref{eq:rho} and Lemma \ref{lem:kx}, we must have $x_t\to \rho$. Therefore, set $x_t=\rho+C\sigma^t$ for some constants $C>0$ and $\sigma\in (0,1)$. Then \eqref{eq:xineq1} clearly holds. To check \eqref{eq:xineq2}, we compute
\begin{align}
    x_t+\frac{\rho}{x_{t+1}}-1-\rho&=\rho+C\sigma^t+\frac{\rho}{\rho+C\sigma^{t+1}}-1-\rho \notag \\
   &=\frac{C\sigma^t}{\rho+C\sigma^{t+1}}(\rho+C\sigma^{t+1}-\sigma). \label{eq:xineq2diff2}
\end{align}
Thus \eqref{eq:xineq2} holds if $\sigma\le \rho$. To check whether the equilibrium is bubbly, using \eqref{eq:kdx_dp} and \eqref{eq:xineq2diff2}, we compute
\begin{equation}
    \frac{d_t}{p_t}=\frac{C\sigma^{t-1}}{\rho+C\sigma^t}(\rho+C\sigma^t-\sigma)\frac{1}{1-\frac{\rho}{\rho+C\sigma^t}}=\frac{\rho-\sigma+C\sigma^t}{\sigma}. \label{eq:dp4}
\end{equation}
Since $\sigma\le \rho$, we have $\sum_{t=1}^\infty d_t/p_t<\infty$ if and only if $\sigma=\rho$. Under this condition, by Lemma \ref{lem:bubble} the equilibrium is bubbly. Since $d_t/p_t=C\rho^{t-1}$ and $p_t$ is bounded above, we have $\sum_{t=1}^\infty d_t<\infty$ and Assumption \ref{asmp:D} holds. Furthermore, $k_t\to k^*$ implies $0\le b_t\le p_t\to 0$ by \eqref{eq:xdef_p}, so the equilibrium is asymptotically bubbleless. Therefore, we obtain the following result.

\begin{exmp}[Bubbly but asymptotically bubbleless equilibrium]\label{exmp:asym_bubbleless}
Let everything be as in Lemma \ref{lem:CD}. Let $\alpha<\frac{\beta}{1+\beta}$ so that $\rho$ in \eqref{eq:rho} satisfies $\rho\in (0,1)$. For any $C>0$, let $x_t=\rho+C\rho^t$ and define $\set{(k_{t+1},p_t,d_{t+1})}_{t=0}^\infty$ by \eqref{eq:xdef}. Then $\set{(k_t,p_t)}_{t=0}^\infty$ is bubbly but asymptotically bubbleless.
\end{exmp}

Because all assumptions of Theorem \ref{thm:eqset} are satisfied, and Example \ref{exmp:asym_bubbleless} provides a bubbly but asymptotically bubbleless equilibrium, case \ref{item:eqset3} in Theorem \ref{thm:eqset} is possible. Figure \ref{fig:asymbubbleless} shows the time path of $\set{(k_t,p_t,d_t)}$ for a numerical example. We set $G=1$ and $A=1/(\beta(1-\alpha))$ so that there is no growth and the steady-state capital is normalized to 1.

\begin{figure}[!htb]
\centering
\includegraphics[width=0.7\linewidth]{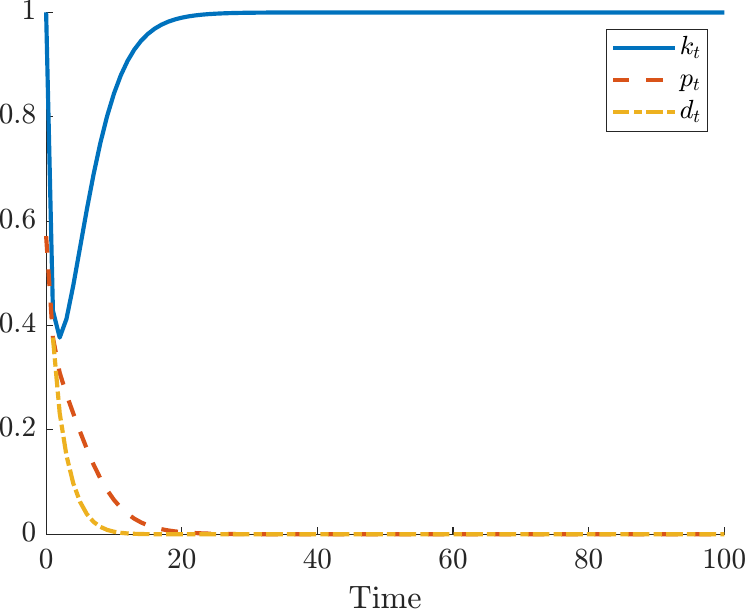}
\caption{Bubbly but asymptotically bubbleless equilibrium.}\label{fig:asymbubbleless}
\caption*{\footnotesize Note: Parameter values in Example \ref{exmp:asym_bubbleless} are $G=1$, $\alpha=1/3$, $\beta=2/3$, $A=1/(\beta(1-\alpha))$, $C=1$, and $k_0=1$.}
\end{figure}

The exact reverse-engineering construction in this section exploits two special features of the log-Cobb-Douglas specification: saving is a constant fraction of wage income, and the capital law of motion is a power recursion. These features make it possible to choose the auxiliary sequence $\set{x_t}$ and recover dividends in closed form. The general results in \S\ref{sec:main} do not rely on these functional forms; the purpose of the examples is instead to show constructively that each case in Theorem \ref{thm:eqset} can occur. We do not claim that the same closed-form formulas extend to general preferences and technologies.

\appendix

\section{Proof of \S\ref{sec:model} results}\label{sec:proof}

\subsection{Proof of Lemma \ref{lem:eq}}

Take any rational expectations equilibrium \eqref{eq:seq1}.

We first show that without loss of generality, the no-arbitrage condition \eqref{eq:R} holds. If $P_t=0$, then it must be $P_{t+1}+D_{t+1}=0$, for otherwise the agent can increase old consumption $c_{t+1}^o$ at no cost by letting $x_t\to\infty$ and increase utility, which violates optimality. Then \eqref{eq:R} holds. Suppose $P_t>0$ and let $R_{t+1}'=(P_{t+1}+D_{t+1})/P_t$ be the gross return on the asset. To simplify the notation, let us suppress time subscripts. If $R'<R$, the asset return is dominated by the capital return. Since asset holdings $x$ is unrestricted, agents can reduce $x$ by $\Delta x$ and increase $i$ by $P\Delta x$, which leaves young consumption $c^y$ unchanged but increases old consumption $c^o$ by $(R-R')P\Delta x>0$, which violates optimality. Therefore, it must be $R'\ge R$. If $R'>R$, the capital return is dominated by the asset return, so agents choose $i=i_t=0$ and hence $K=K_{t+1}=0$. Homogeneity of $F$ implies zero profits $0=F(0,L)-R\cdot 0-wL$, while profit maximization implies $0\ge F(K,L)-RK-wL$ for all $K\ge 0$. Since $R'>R$, these two conditions also hold with $R'$ instead of $R$. Thus we may simply redefine the capital rent as $R'$ without changing the equilibrium allocation. Therefore, we may assume $R'=R$, and the no-arbitrage condition \eqref{eq:R} holds.

Under the no-arbitrage condition \eqref{eq:R}, if we define $s=i+Px$, the budget constraints \eqref{eq:budget} reduce to $c^y+s=w$ and $c^o=Rs$. Therefore, the utility maximization problem reduces to maximizing $U(w-s,Rs)$, which is \eqref{eq:UMP}. 

Profit maximization is the same in Definition \ref{defn:eq} and Lemma \ref{lem:eq}.

Finally, using the capital market clearing condition \eqref{eq:k_clear}, asset market clearing condition \eqref{eq:x_clear}, and the definition $s_t=i_t+P_tx_t$, we obtain
\begin{equation*}
    N_ts_t=N_t(i_t+P_tx_t)=N_ti_t+P_tN_tx_t=K_{t+1}+P_t,
\end{equation*}
which is \eqref{eq:s_clear}.

Conversely, let the sequence $\set{(P_t,R_{t+1},w_t,s_t,K_t)}_{t=0}^\infty$ satisfy the conditions in Lemma \ref{lem:eq}. Define $c_t^y=w_t-s_t$, $c_{t+1}^o=R_{t+1}s_t$, $i_t=K_{t+1}/N_t$, $x_t=1/N_t$, and $L_t=N_t$. Using the budget constraints \eqref{eq:budget}, it is clear that utility maximization, profit maximization, labor market clearing, capital market clearing \eqref{eq:k_clear}, and asset market clearing \eqref{eq:x_clear} hold. Therefore, it suffices to show the commodity market clearing condition \eqref{eq:c_clear}. 

By the definition of $c_t^y$, $L_t=N_t$, and asset market clearing \eqref{eq:s_clear}, we obtain
\begin{equation*}
    N_tc_t^y=N_t(w_t-s_t)=w_tL_t-(K_{t+1}+P_t).
\end{equation*}
Similarly, the definition of $c_t^o$ and \eqref{eq:s_clear} imply
\begin{equation*}
    N_{t-1}c_t^o=N_{t-1}R_ts_{t-1}=R_t(K_t+P_{t-1}).
\end{equation*}
Taking the sum and using the definition $i_t=K_{t+1}/N_t$, we obtain
\begin{align*}
    N_t(c_t^y+i_t)+N_{t-1}c_t^o&=w_tL_t-(K_{t+1}+P_t)+K_{t+1}+R_t(K_t+P_{t-1})\\
    &=(R_tK_t+w_tL_t)+R_tP_{t-1}-P_t=F_t(K_t,L_t)+D_t,
\end{align*}
where the last equality follows from zero profit and the no-arbitrage condition \eqref{eq:R}. Therefore, the commodity market clearing condition \eqref{eq:c_clear} holds. \hfill \qedsymbol

\subsection{Proof of Proposition \ref{prop:exist}}

The idea of the proof is similar to \citet{BalaskoShell1980} and \citet{Wilson1981}, who consider endowment economies, as well as \citet{LeVanPham2016}, who consider an infinite-horizon production economy with capital and productive asset. For each $T\in \N$, we define a $T$-equilibrium as follows.

\begin{defn}\label{defn:Teq}
The sequence \eqref{eq:seq1} is a \emph{$T$-equilibrium} if \ref{item:eq1} individual optimization holds for $t=0,1,\dots,T-1$, \ref{item:eq2} profit maximization holds for $t=0,1,\dots,T$, \ref{item:eq4} labor market clearing holds for $t=0,1,\dots,T$, and \ref{item:eq5} capital and asset market clearing holds for $t=0,1,\dots,T-1$.
\end{defn}

We first prove the existence of a $T$-equilibrium. To this end, for each $T\in \N$, we define a $T$-truncated Arrow-Debreu economy $\cE_T$ as follows.
\begin{itemize}
    \item Time is denoted by $t=-1,0,1,\dots,T$. For each $t\ge -1$, there is a consumption good with (date 0) price $q_t$. For each $t\ge 0$, there is labor service with (date 0) price $\omega_t$.
    \item For each $t$, there are homogeneous agents with population $N_t>0$ with the following preferences and endowments.
    \begin{enumerate*}
        \item Each agent in generation $t=-1$ has utility $u_0(c_0^o)=c_0^o$ over date 0 consumption $c_0^o$ and is endowed with $K_0/N_{-1}$ units of date $t=-1$ good and $D_t/N_{-1}$ units of date $t\ge 0$ good.
        \item Each agent in generation $t=0,1,\dots,T-1$ has utility $U_t(c_t^y,c_{t+1}^o)$ over date $(t,t+1)$ consumption $(c_t^y,c_{t+1}^o)$ and is endowed with a unit of date $t$ labor service.
        \item Each agent in generation $T$ has utility $u_T(c_T^y)=c_T^y$ over date $T$ consumption $c_T^y$ and is endowed with a unit of date $T$ labor service.
    \end{enumerate*}
    \item For each $t\ge 0$, there is a firm that uses date $t-1$ consumption good $K_t$ and date $t$ labor $L_t$ as inputs to produce the date $t$ consumption good. Let $F_t(K_t,L_t)$ be the production function.
\end{itemize}
A competitive equilibrium of $\cE_T$ consists of sequences of prices $\set{q_t}_{t=-1}^T$, $\set{\omega_t}_{t=0}^T$, consumption $\set{(c_t^y,c_t^o)}_{t=0}^T$, and inputs $\set{(K_t,L_t)}_{t=0}^T$ such that,
\begin{enumerate}
    \item (Utility maximization) For each $t=0,\dots,T-1$, $(c_t^y,c_{t+1}^o)$ maximizes utility $U_t$ subject to the budget constraint
    \begin{equation*}
        q_tc_t^y+q_{t+1}c_{t+1}^o\le \omega_t.
    \end{equation*}
    Furthermore, $q_0c_0^o=q_{-1}K_0/N_{-1}+\sum_{t=0}^Tq_tD_t/N_{-1}$ and $q_Tc_T^y=\omega_T$.
    \item (Profit maximization) For each $t=0,\dots,T$, firm $t$ maximizes the profit
    \begin{equation*}
        q_tF_t(K_t,L_t)-q_{t-1}K_t-\omega_tL_t.
    \end{equation*}
    \item (Commodity market clearing) For each $t=0,\dots,T$, the commodity market clears:
    \begin{equation*}
        N_tc_t^y+N_{t-1}c_t^o+K_{t+1}=F_t(K_t,L_t)+D_t,
    \end{equation*}
    where $K_{T+1}=0$.
    \item (Labor market clearing) For each $t=0,1,\dots,T$, the labor market clears: $L_t=N_t$.
\end{enumerate}
Note that since $F_t$ is homogeneous of degree 1, the maximized profit is zero, so we do not need to specify the ownership of firms.

\begin{lem}\label{lem:Teq}
A $T$-equilibrium exists.
\end{lem}

\begin{proof}
By standard results \citep{ArrowDebreu1954}, the $T$-truncated Arrow-Debreu economy $\cE_T$ has a competitive equilibrium. The strict monotonicity of $U_t$ implies $q_t>0$ for all $t$. Define $R_t=q_{t-1}/q_t>0$ and $w_t=\omega_t/q_t$. Define $\set{P_t}_{t=0}^T$ recursively by $P_T=0$ and $P_{t-1}=(P_t+D_t)/R_t$. Finally, define capital investment by $i_t=K_{t+1}/N_t$ using \eqref{eq:k_clear} and asset holdings $x_t=1/N_t$ using \eqref{eq:x_clear}. If we define variables at time $t>T$ arbitrarily, this Arrow-Debreu equilibrium is part of the $T$-equilibrium in Definition \ref{defn:Teq}.
\end{proof}

\begin{proof}[Proof of Proposition \ref{prop:exist}]
We first bound equilibrium quantities. Define the sequence $\set{\bar{K}_t}_{t=0}^\infty$ by $\bar{K}_0=K_0>0$ and $\bar{K}_{t+1}=F_t(\bar{K_t},N_t)$ for $t\ge 0$. Since $F_t$ has positive partial derivatives, we have $\bar{K}_t>0$ for all $t$. If an equilibrium exists, market clearing \eqref{eq:s_clear} and the homogeneity of $F_t$ imply
\begin{equation*}
    0\le K_{t+1}+P_t=N_ts_t\le w_tN_t\le R_tK_t+w_tN_t=F_t(K_t,N_t).
\end{equation*}
By induction, we must have
\begin{align*}
    &0\le K_t\le \bar{K}_t, & 
    &0\le P_t\le F_t(\bar{K}_t,N_t)\eqqcolon \bar{P}_t,\\
    &0\le s_t\le F_t(\bar{K}_t,N_t)/N_t\eqqcolon \bar{f}_t, &
    &0\le w_t\le F_t(\bar{K}_t,N_t)/N_t\eqqcolon \bar{f}_t.
\end{align*}
Young's budget constraint \eqref{eq:budget_young} implies the bound $0\le c_t^y\le w_t\le \bar{f}_t$. The commodity market clearing condition \eqref{eq:c_clear} implies the bound
\begin{equation}
    0\le c_t^o\le (F_t(\bar{K}_t,N_t)+D_t)/N_{t-1}\eqqcolon \bar{c}_t^o. \label{eq:co_ub}
\end{equation}
Finally, we bound the date 0 prices $q_t,\omega_t$. Normalize $q_0=1$. Since $F_t$ is concave and continuously differentiable, $F_{t,K}(K,L)>0$ is continuous and decreasing in $K$. By profit maximization and $K_t\le \bar{K}_t$, we have
\begin{equation*}
    \frac{q_{t-1}}{q_t}=R_t\ge F_{t,K}(K_t,N_t)\ge F_{t,K}(\bar{K}_t,N_t)\eqqcolon \bar{R}_t.
\end{equation*}
Therefore, $q_t=1/\prod_{s=1}^tR_s\le 1/\prod_{s=1}^t\bar{R}_s\eqqcolon \bar{q}_t$. Obviously,
\begin{equation*}
    0\le \omega_t=q_tw_t\le \bar{q}_t\bar{f}_t\eqqcolon \bar{\omega}_t.
\end{equation*}

Collect the equilibrium quantities as $\xi_t=(q_t,\omega_t,c_t^y,c_t^o,K_t,P_t)$ and set $\xi=(\xi_t)$. Define the nonempty compact set
\begin{equation*}
    \Xi_t=[0,\bar{q}_t]\times [0,\bar{\omega}_t]\times [0,\bar{f}_t]\times [0,\bar{c}_t^o]\times [0,\bar{K}_t]\times [0,\bar{P}_t]\subset \R^6
\end{equation*}
and $\Xi=\prod_{t=0}^\infty \Xi_t$, where we endow $\Xi$ with the product topology induced by the Euclidean topology on $\Xi_t\subset \R^6$. By Tychonoff's theorem, $\Xi$ is compact. Let $E_T\subset \Xi$ be the set of $\xi=(\xi_t)$ that induces a $T$-equilibrium. By Lemma \ref{lem:Teq}, $E_T$ is nonempty. Standard arguments show that $E_T$ is closed. Since $E_1\supset E_2\supset\dotsb$ and $\Xi$ is compact, we have $E\coloneqq \bigcap_{t=1}^\infty E_t\neq \emptyset$. 

Take $\xi=(\xi_t)\in E$. Since $\xi$ also induces a $T$-equilibrium for all $T$, it must be $q_t>0$, for otherwise utility maximization does not hold. Define $R_t=q_{t-1}/q_t>0$, $w_t=\omega_t/q_t$, and $s_t=w_t-c_t^y$. By the proof of Lemma \ref{lem:Teq}, the no-arbitrage condition \eqref{eq:R} holds. Therefore, all conditions in Lemma \ref{lem:eq} hold, so we have a rational expectations equilibrium.
\end{proof}

\subsection{Proof of Lemma \ref{lem:k}}

To simplify notation, we suppress time subscripts. Using the homogeneity of $F$ and the definition of $f$, the profit \eqref{eq:profit} can be written as
\begin{equation*}
    F(K,L)-RK-wL=L(f(k)-Rk-w)\eqqcolon L\pi(k),
\end{equation*}
where $k=K/L$. Then $\pi'(k)=f'(k)-R\to\infty$ as $k\downarrow 0$, so the optimal $k$ must be $k>0$ and satisfies $R=f'(k)\ge f'(\infty)$ due to the concavity of $f$.

Using $k=K/L$, the profit may also be written as
\begin{equation*}
    F(K,L)-RK-wL=K\left(\frac{1}{k}f(k)-R-\frac{w}{k}\right)=K\frac{\pi(k)}{k}.
\end{equation*}
Since $k>0$, profit maximization implies
\begin{equation*}
    0=\frac{\diff}{\diff k}\frac{\pi(k)}{k}=\frac{(f'(k)-R)k-(f(k)-Rk-w)}{k^2} \iff w=f(k)-kf'(k).
\end{equation*}
Clearly $w=F_L(K,L)>0$. Using market clearing \eqref{eq:s_clear} and the homogeneity of $F$, we obtain
\begin{equation}
    \max\set{K_{t+1},P_t}\le  K_{t+1}+P_t=N_ts_t\le w_tN_t\le R_tK_t+w_tN_t=F(K_t,N_t). \label{eq:Kbound}
\end{equation}
Dividing both sides by $N_t$, we obtain $Gk_{t+1}\le f(k_t)$ and $p_t\le f(k_t)$.

Finally, we show that $\set{k_t}$ is bounded.
By Assumption \ref{asmp:f}, $f$ is increasing, concave, and $f'(\infty)<G$. Therefore, we can take constants $a\in (0,1)$ and $b\ge 0$ such that
\begin{equation*}
    0\le k_{t+1}\le \frac{1}{G}f(k_t)\le ak_t+b.
\end{equation*}
Iterating this inequality yields
\begin{equation*}
    k_t\le a^t\left(k_0-\frac{b}{1-a}\right)+\frac{b}{1-a}.
\end{equation*}
Letting $t\to\infty$, we obtain $\limsup_{t\to\infty}k_t\le b/(1-a)$, so $\set{k_t}$ is bounded. \hfill \qedsymbol

\subsection{Proof of Lemma \ref{lem:s}}

The first-order condition for the utility maximization problem is
\begin{equation*}
    \Phi(s,w,R)\coloneqq -u'(w-s)+Rv'(Rs)=0.
\end{equation*}
By assumption, $\Phi$ is continuous and strictly decreasing in $s$, strictly increasing in $w$, and increasing in $R$. Since $u'(0)=v'(0)=\infty$, we have $\Phi(0,w,R)=\infty$, and $\Phi(w,w,R)=-\infty$. By the intermediate value theorem, there exists a unique $s=s(w,R)\in (0,w)$ such that $\Phi(s,w,R)=0$, which achieves the unique maximum of $U(w-s,Rs)=u(w-s)+v(Rs)$. Since $u',v'$ are continuous, so is $s$.

To show the monotonicity of $s$, let $w^1<w^2$ and $R^1\le R^2$. Fixing $R$, we obtain
\begin{equation*}
    \Phi(s(w^2,R),w^2,R)=0=\Phi(s(w^1,R),w^1,R)< \Phi(s(w^1,R),w^2,R).
\end{equation*}
Since $\Phi$ is strictly decreasing in $s$, we obtain $s(w^2,R)>s(w^1,R)$, so $s$ is strictly increasing in $w$. Similarly, fixing $w$, we obtain
\begin{equation*}
    \Phi(s(w,R^2),w,R^2)=0=\Phi(s(w,R^1),w,R^1)\le \Phi(s(w,R^1),w,R^2).
\end{equation*}
Since $\Phi$ is strictly decreasing in $s$, we obtain $s(w,R^1)\le s(w,R^2)$, so $s$ is increasing in $R$. \hfill \qedsymbol

\subsection{Proof of Lemma \ref{lem:g}}
Let $\Phi(x,k,p)$ be the left-hand side of \eqref{eq:xeq}. By Lemma \ref{lem:s} and Assumption \ref{asmp:f}, $x\mapsto s(f(k)-kf'(k),f'(x))$ is decreasing. Therefore, $\Phi$ is strictly increasing in $x$, so \eqref{eq:xeq} has at most one solution denoted by $x=g(k,p)$.

\medskip
\noindent
\ref{item:g1} Fix $\epsilon>0$. If $x\in (0,\epsilon)$, then
\begin{equation*}
    \Phi(x,k,0)=Gx-s(f(k)-kf'(k),f'(x))\le Gx-s(f(k)-kf'(k),f'(\epsilon)).
\end{equation*}
Therefore,
\begin{equation*}
    \lim_{x\downarrow 0}\Phi(x,k,0)\le -s(f(k)-kf'(k),f'(\epsilon))<0.
\end{equation*}
Similarly, since $s(w,R)<w$, we obtain
\begin{equation}
    \Phi(x,k,0)=Gx-s(f(k)-kf'(k),f'(x))>Gx-(f(k)-kf'(k)), \label{eq:xk0}
\end{equation}
so $\lim_{x\to\infty}\Phi(x,k,0)=\infty$. Therefore, $x=g(k,0)$ exists, so $(k,0)\in \dom g$. Furthermore, setting $x=k$ in \eqref{eq:xk0}, dividing by $k$, and letting $k\to\infty$, we obtain
\begin{equation*}
    \frac{\Phi(k,k,0)}{k}>G-\frac{f(k)}{k}+f'(k)\to G-f'(\infty)+f'(\infty)=G>0
\end{equation*}
by l'H\^opital's theorem. Therefore, $g(k,0)<k$ for large enough $k$.

\medskip
\noindent
\ref{item:g2} To show the continuity of $g$, suppose $(k_n,p_n)\to (k,p)$. Any root $x_n\coloneqq g(k_n,p_n)$ is locally bounded because $Gx_n+p_n=s(\omega(k_n),f'(x_n))\le \omega(k_n)$. Take any convergent subsequence $x_{n_j}\to x$. Then the continuity of $\Phi$ implies $\Phi(x,k,p)=0$, and uniqueness implies $x=g(k,p)$. Therefore, since the limit is unique, we have $g(k_n,p_n)=x_n\to x=g(k,p)$, so $g$ is continuous.

Since $(f(k)-kf'(k))'=-kf''(k)>0$, by Lemma \ref{lem:s} $\Phi$ is strictly decreasing in $k$. Clearly $\Phi$ is strictly increasing in $p$. The rest of the proof is the same as Lemma \ref{lem:s}.

\medskip
\noindent 
\ref{item:g3} Let $(k,p)\in \dom g$, $k'\ge k$, and $0\le p'\le p$. By the definition of $g$ and the monotonicity of $\Phi$, we have
\begin{align*}
    \Phi(g(k,p),k',p')\le \Phi(g(k,p),k,p)=0.
\end{align*}
Since $s(w,R)<w$, for $w'=f(k')-k'f'(k')$ we obtain
\begin{equation*}
    \Phi(x,k',p')=Gx+p'-s(w',f'(x))\ge Gx-w'\to \infty
\end{equation*}
as $x\to\infty$, so there exists a unique $x\ge g(k,p)$ such that $\Phi(x,k',p')=0$. Therefore, $(k',p')\in \dom g$ and \eqref{eq:gineq} holds. \hfill \qedsymbol

\subsection{Proof of Proposition \ref{prop:p0}}

\ref{item:p0_interval} Proposition \ref{prop:exist} implies $\cP_0\neq\emptyset$. Suppose $p_0^1,p_0^2\in \cP_0$ with $p_0^1<p_0^2$ and let $p_0\in [p_0^1,p_0^2]$. Let $\set{(k_t^j,p_t^j)}_{t=0}^\infty$ be detrended capital and asset price corresponding to $p_0^j$. By assumption, we have $k_0^j=k_0$.

Let us show by induction that there exists a unique sequence $\set{(k_t,p_t)}_{t=0}^T$ satisfying \eqref{eq:system}, $p_t\in [p_t^1,p_t^2]$, and $k_t\in [k_t^2,k_t^1]$. If $T=0$, the claim is trivial because $p_0\in [p_0^1,p_0^2]$ and  $k_0=k_0^1=k_0^2$ are given. Suppose the claim holds up to $T-1$ and consider $T$. By the induction hypothesis, there exist unique $p_t\in [p_t^1,p_t^2]$ and $k_t\in [k_t^2,k_t^1]$ for $t\le T-1$. Using \eqref{eq:system_k} and \eqref{eq:gineq}, we obtain
\begin{equation*}
    k_T^2=g(k_{T-1}^2,p_{T-1}^2)\le g(k_{T-1},p_{T-1})\le g(k_{T-1}^1,p_{T-1}^1)=k_T^1,
\end{equation*}
so $k_T\in [k_T^2,k_T^1]\subset (0,\infty)$ is uniquely defined by \eqref{eq:system_k}. Letting $R_t^j=f'(k_t^j)$ and $R_t=f'(k_t)$, $k_T\in [k_T^2,k_T^1]$ and $f''<0$ imply $R_T^1\le R_T\le R_T^2$. Therefore, \eqref{eq:system_p} implies
\begin{equation*}
    p_T^1=\frac{R_T^1}{G}p_{T-1}^1-d_T\le \frac{R_T}{G}p_{T-1}-d_T\le \frac{R_T^2}{G}p_{T-1}^2-d_T=p_T^2,
\end{equation*}
so $p_T\in [p_T^1,p_T^2]\subset [0,\infty)$ is uniquely defined by \eqref{eq:system_p}. Thus the claim holds for $T$ as well. By induction, $\set{(k_t,p_t)}_{t=0}^\infty$ satisfies the equilibrium system \eqref{eq:system}, so $p_0\in \cP_0$. Therefore, $\cP_0$ is an interval.

To show that $\cP_0$ is compact, define the sequence $\set{\bar{k}_t}_{t=0}^\infty$ by $\bar{k}_0=k_0>0$ and $\bar{k}_{t+1}=f(\bar{k}_t)/G>0$. Using the market clearing condition \eqref{eq:s_clear} and the fact that the savings function satisfies
\begin{equation*}
    s(w,R)<w=f(k)-kf'(k)<f(k),
\end{equation*}
it follows that in any equilibrium, we have $k_t\in [0,\bar{k}_t]$ and $p_t\in [0,f(\bar{k}_t)]$. Let $\ubar{p}_0=\inf \cP_0$ and $\bar{p}_0=\sup \cP_0\le f(k_0)$. To show that $\cP_0$ is compact, it suffices to show that $\ubar{p}_0,\bar{p}_0\in \cP_0$.

By the definition of $\ubar{p}_0$, we can take a decreasing sequence $\set{p_0^n}_{n=1}^\infty\subset \cP_0$ such that $p_0^n\downarrow \ubar{p}_0$. Let $\set{(k_t^n,p_t^n)}_{t=0}^\infty$ be the corresponding path. By the above proof, $k_t^n\in [0,\bar{k}_t]$ is increasing in $n$ and $p_t^n\in [0,f(\bar{k}_t)]$ is decreasing in $n$, so they converge to some $(k_t,p_t)$ with $p_0=\ubar{p}_0$. Since $k_t^n>0$ is increasing in $n$, we have $k_t>0$. Clearly $\set{(k_t,p_t)}_{t=0}^\infty$ satisfies the equilibrium system \eqref{eq:system}, so we have an equilibrium. Therefore, $\ubar{p}_0\in \cP_0$.

The argument for $\bar{p}_0$ is similar, except that now $k_t^n$ is decreasing in $n$ and could converge to 0. To show that this never occurs, let $k_t^n\to k_t$, suppose to the contrary that $k_t=0$, and that $t$ is the smallest such $t$. Then $w_{t-1}^n\to w_{t-1}=f(k_{t-1})-k_{t-1}f'(k_{t-1})>0$ (because $k_{t-1}>0$ and by Lemma \ref{lem:k}) and $R_t^n=f'(k_t^n)\to f'(0)=\infty$ as $n\to \infty$. Choose $w,R>0$ such that $w_{t-1}^n>w$ and $R_t^n>R$ for large $n$. Then the old's consumption becomes
\begin{equation}
    c_t^{o,n}=R_t^ns(w_{t-1}^n,R_t^n)\ge R_t^ns(w,R)\to \infty \label{eq:co_inf}
\end{equation}
by Assumption \ref{asmp:s}, contradicting the uniform bound in \eqref{eq:co_ub}. Therefore, $k_t>0$ for all $t$, we have an equilibrium, and $\bar{p}_0\in \cP_0$.

\medskip
\noindent
\ref{item:p0_monotone}
The proof of $k_t>k_t'$ and $p_t<p_t'$ follows from the proof of \ref{item:p0_interval} and using the strict monotonicity of $g$ established in Lemma \ref{lem:g}\ref{item:g2}. Since $f''<0$, and $(f(k)-kf'(k))'=-kf''(k)>0$, we obtain $R_t=f'(k_t)<f'(k_t')=R_t'$ and
\begin{equation*}
    w_t=f(k_t)-k_tf'(k_t)>f(k_t')-k_t'f'(k_t')=w_t'.
\end{equation*}
Since $R_t<R_t'$, the fundamental values satisfy
\begin{equation*}
    V_t\coloneqq \sum_{s=1}^\infty \frac{D_{t+s}}{R_{t+1}\dotsb R_{t+s}}\ge \sum_{s=1}^\infty \frac{D_{t+s}}{R_{t+1}'\dotsb R_{t+s}'}\eqqcolon V_t'.
\end{equation*}
Dividing both sides by $G^t$ yields $v_t\ge v_t'$. Finally, $p_t<p_t'$ and $v_t\ge v_t'$ imply $b_t=p_t-v_t<p_t'-v_t'=b_t'$. \hfill \qedsymbol

\section{Proof of \S\ref{sec:main} results}

\subsection{Proof of Lemma \ref{lem:diamond}}
\noindent \ref{item:diamond_converge} Let $m(k)\coloneqq g(k,0)$. By Lemma \ref{lem:g}, $m$ is continuous, strictly increasing, and $m(k)<k$ for large enough $k$. Regardless of $k_1^*=m(k_0^*)\le k_0^*$ or $k_1^*>k_0^*$, by induction $\set{k_t^*}$ is monotone and hence converges to some $k^*\in [0,\infty]$. Since $m(k)<k$ for large enough $k$, it cannot be $k^*=\infty$. By continuity, it must be $k^*\in \set{0}\cup \cK^*$.

Suppose $g(k,0)>k$ for small enough $k>0$ and let $h(k)\coloneqq g(k,0)-k$. By assumption, there exists $a>0$ such that $h(k)>0$ for $k\in (0,a)$. By Lemma \ref{lem:g}\ref{item:g1}, there exists $b>0$ such that $h(k)<0$ for $k>b$. Since $h$ is continuous, by the intermediate value theorem there exists $k\in [a,b]$ such that $h(k)=0$, so $\cK^*$ is nonempty. Since $\cK^*=\set{k\in [a,b]:h(k)=0}$ is a closed subset of the compact set $[a,b]$, it is compact. Finally, we show $k^*>0$ and hence $k^*\in \cK^*$. If $\set{k_t^*}$ is increasing, we have $k^*\ge k_0>0$. Suppose $\set{k_t^*}$ is decreasing. If $k^*=0$, then $k_t^*$ becomes arbitrarily small, so $k_{t+1}^*=g(k_t^*,0)>k_t^*$, contradicting monotonicity. Therefore $k^*>0$.

\medskip
\noindent \ref{item:diamond_ub} By definition, $k_0^*=k_0$. If $k_t\le k_t^*$, since $p_t\ge 0$, Lemma \ref{lem:g}\ref{item:g3} implies $k_{t+1}=g(k_t,p_t)\le g(k_t^*,0)=k_{t+1}^*$. By induction, $k_t\le k_t^*$ for all $t$. \hfill \qedsymbol

\subsection{Proof of Theorem \ref{thm:bubbleless}}

Take the $T$-equilibrium established in Lemma \ref{lem:Teq}. Let $R_t=q_{t-1}/q_t>0$ and define $\set{P_t}_{t=0}^T$ recursively by $P_T=0$ and $P_{t-1}=(P_t+D_t)/R_t$. Let $k_t=K_t/G^t$ and $p_t=P_t/G^t$. Then \eqref{eq:system} holds for $t=0,\dots,T-1$. By the proof of Proposition \ref{prop:exist}, for fixed $t$, the sequence $\set{(k_t^T,p_t^T)}_{T=t}^\infty$ is uniformly bounded, namely it belongs to the compact set $[0,\bar{K}_t/G^t]\times [0,\bar{P}_t/G^t]$. Therefore, applying the diagonal argument, we can take a subsequence $T_1<T_2<\dotsb$ such that for each $t$, we have $(k_t^{T_n},p_t^{T_n})\to(k_t,p_t)$ as $n\to\infty$. By the same argument as in the proof of Proposition \ref{prop:p0} (see around \eqref{eq:co_inf}), we have $k_t>0$. Clearly, $\set{(k_t,p_t)}_{t=0}^\infty$ satisfies the equilibrium system \eqref{eq:system}, so we have an equilibrium. Let $R_t=f'(k_t)$ and define the date 0 price $q_t\coloneqq 1/\prod_{s=1}^t R_s$. Similarly, let $q_t^{T_n}$ be the date 0 price in the $T_n$-equilibrium and $q_t^*\coloneqq 1/\prod_{s=1}^tR_s^*$. By the same induction as in Lemma \ref{lem:diamond}\ref{item:diamond_ub}, every $T$-equilibrium satisfies $k_t^T\le k_t^*$ and $R_t^T\ge R_t^*$ for $t=0,\dots,T$, so $q_t^T\le q_t^*$. Furthermore, $q_t^{T_n}\to q_t\le q_t^*$ as $n\to\infty$.

For every fixed $T\in \N$, and all sufficiently large $n$ with $T_n>T$, we have
\begin{equation*}
    0\le p_0^{T_n}-\sum_{t=1}^T q_t^{T_n}D_t=\sum_{t=T+1}^{T_n}q_t^{T_n}D_t\le \sum_{t=T+1}^\infty q_t^*D_t.
\end{equation*}
Letting $n\to\infty$, we obtain
\begin{equation*}
    0\le p_0-\sum_{t=1}^T q_tD_t\le \sum_{t=T+1}^\infty q_t^*D_t.
\end{equation*}
Letting $T\to\infty$ and using \eqref{eq:bubbleless_cond}, we obtain $p_0=\sum_{t=1}^\infty q_tD_t\eqqcolon v_0$, so $p_0$ is bubbleless. By Corollary \ref{cor:unique_bubbleless}, the bubbleless equilibrium is unique.

We next show that $G_d<R^*\coloneqq f'(k^*)$ implies \eqref{eq:bubbleless_cond}. Since $\limsup_{t\to\infty}D_t^{1/t}=G_d$ and $R_t^*\to R^*$, we can take $\epsilon>0$, $R<R^*$, and $T>0$ such that
\begin{equation*}
    D_t^{1/t}<G_d+\epsilon<R<R_t^*
\end{equation*}
for $t\ge T$. Then for $t>T$, the $t$-th term in \eqref{eq:bubbleless_cond} can be bounded above by
\begin{equation*}
    \frac{(G_d+\epsilon)^t}{R_1^*\dotsb R_T^* R^{t-T}}=\frac{(G_d+\epsilon)^T}{R_1^*\dotsb R_T^*}\left(\frac{G_d+\epsilon}{R}\right)^{t-T},
\end{equation*}
which is summable. \hfill \qedsymbol

\subsection{Proof of Proposition \ref{prop:longrun}}

We need several lemmas to prove Proposition \ref{prop:longrun}. The following lemma shows that capital converges whenever $p_t\to 0$.

\begin{lem}\label{lem:kconverge}
Suppose Assumptions \ref{asmp:G}--\ref{asmp:s} hold. If $\lim_{t\to\infty}p_t=0$, then $k_t\to k^*\in \set{0}\cup \cK^*$.
\end{lem}

\begin{proof}
By Lemma \ref{lem:k}, $\set{k_t}$ is bounded. If $k_t\to k^*$, since $p_t\to 0$, it must be $k^*\in \set{0}\cup \cK^*$. Therefore, it suffices to show that $\set{k_t}$ converges.

Let $\ubar{k}\coloneqq \liminf_{t\to\infty}k_t\ge 0$  and $\bar{k}\coloneqq \limsup_{t\to\infty}k_t<\infty$. Suppose, toward a contradiction, that $\ubar{k}<\bar{k}$ and take any $k\in (\ubar{k},\bar{k})$. Suppose that $g(k,0)>k$. By the continuity of $g$, we can take $\eta>0$ such that $g(k,p)>k$ for $p\in [0,\eta]$. Since $p_t\to 0$, there exists $T$ such that $p_t\le \eta$ for $t\ge T$. Since $\limsup_{t\to\infty}k_t>k$, we can  choose $t\ge T$ such that $k_t>k$. By the monotonicity of $g$, we have
\begin{equation*}
    k_{t+1}=g(k_t,p_t)\ge g(k,\eta)>k.
\end{equation*}
By induction, we have $k_t>k$ for all sufficiently large $t$, so $\liminf_{t\to\infty}k_t\ge k>\ubar{k}$, which is a contradiction.

Therefore, $g(k,0)\le k$ for all $k\in (\ubar{k},\bar{k})$. Take any such $k$. Since $\liminf_{t\to\infty}k_t<k$, we can choose $t$ such that $k_t\le k$. Since $p_t\ge 0$, by the monotonicity of $g$, we have
\begin{equation*}
    k_{t+1}=g(k_t,p_t)\le g(k,0)\le k.
\end{equation*}
By induction, we have $k_t\le k$ for all sufficiently large $t$, so $\limsup_{t\to\infty}k_t\le k<\bar{k}$, which is a contradiction. 
\end{proof}

We next consider the following exhaustive and mutually exclusive cases:
\begin{enumerate*}
    \item (Lemma \ref{lem:Rinc}) $R_t\ge R_{t-1}$ for all $t$;
    \item (Lemma \ref{lem:Rdec1}) $R_t<R_{t-1}$ and $R_t\le G$ for some $t$;
    \item (Lemma \ref{lem:Rdec2}) $R_t<R_{t-1}$ for some $t$, and $R_t>G$ for any such $t$.
\end{enumerate*}

\begin{lem}\label{lem:Rinc}
Suppose Assumptions \ref{asmp:G}--\ref{asmp:D} hold and $R_t\ge R_{t-1}$ for all $t$. Then $k_t\downarrow k\in [0,\infty)$ and $R_t\uparrow R=f'(k)\in (0,\infty]$. Furthermore, exactly one of the following statements holds.
\begin{enumerate}
    \item\label{item:Rinc1} $R>G$, $p_t\to 0$, $k\in \set{0}\cup \cK^*$, and the equilibrium is bubbleless.
    \item\label{item:Rinc2} $R\le G$, $p_t\to 0$, $k\in \cK^*$, and the equilibrium is asymptotically bubbleless.
    \item\label{item:Rinc3} $R=G$, $p_t\to p>0$, $k=g(k,p)>0$, and the equilibrium is asymptotically bubbly.
\end{enumerate}
\end{lem}

\begin{proof}
Since $f''<0$, $R_t=f'(k_t)$ and $R_t\ge R_{t-1}$ imply $k_t\le k_{t-1}$. Hence the monotonic convergence of $\set{k_t}$, $\set{R_t}$ follow. If $\set{p_t}$ converges to $p$, then letting $t\to\infty$ in \eqref{eq:system_k}, we obtain $k=g(k,p)$.

Suppose $R>G$. By Lemma \ref{lem:impossible2}, the equilibrium is unique and bubbleless. Since $\set{R_t}$ is increasing, there exists $t_0$ such that $R_t\ge R_{t_0}>G$ for all $t\ge t_0$. Then for $t\ge t_0$, we have
\begin{equation}
    v_t=\sum_{s\ge 1}\frac{G^s}{R_{t+1}\cdots R_{t+s}}d_{t+s}\le \sum_{s\ge 1}d_{t+s}\to 0 \label{eq:vt_ub}
\end{equation}
as $t\to \infty$ because $\sum_{t}d_t<\infty$. Hence $p_t=v_t\to 0$ and statement \ref{item:Rinc1} holds by Lemma \ref{lem:kconverge}.

Therefore, in what follows, assume $R\le G$. Since $\set{R_t}$ is increasing, we have $R_t\le G$ for all $t$, so \eqref{eq:system_p} implies
\begin{equation*}
    0\le p_t=\frac{R_t}{G}p_{t-1}-d_t\le p_{t-1}
\end{equation*}
for all $t$. Therefore, $\set{p_t}$ converges to some $p\ge 0$. If $p=0$, then $0\le b_t\le p_t\to 0$, so the equilibrium is asymptotically bubbleless. Therefore, \ref{item:Rinc2} holds. If $p>0$, we can take $\ubar{p}>0$ such that $p_t\ge \ubar{p}$ for all $t$. Since
\begin{equation}
    \frac{G}{R_t}=\frac{p_{t-1}}{p_t+d_t}\le \frac{p_{t-1}}{p_t}, \label{eq:GR_ub}
\end{equation}
we can bound the fundamental component from above as
\begin{equation*}
    v_t=\sum_{s=1}^\infty \frac{G^s}{R_{t+1}\dotsb R_{t+s}}d_{t+s}\le \sum_{s=1}^\infty \frac{p_t}{p_{t+s}}d_{t+s}\le \frac{p_t}{\ubar{p}}\sum_{s=1}^\infty d_{t+s}\to \frac{p}{\ubar{p}}\cdot 0=0
\end{equation*}
as $t\to\infty$, so $b_t=p_t-v_t\to p>0$ and the equilibrium is asymptotically bubbly. Letting $t\to\infty$ in \eqref{eq:GR_ub}, we obtain $G/R=1$ and hence $R=G$. Therefore, \ref{item:Rinc3} holds.
\end{proof}

\begin{lem}\label{lem:Rdec1}
Suppose Assumptions \ref{asmp:G}--\ref{asmp:s} hold. If in equilibrium $R_t<R_{t-1}$ and $R_t\le G$ for some $t$, then $\set{(k_t,p_t,R_t)}$ converges to some $(k,0,R)$ satisfying $k\in \cK^*$ and $R=f'(k)$, and the equilibrium is asymptotically bubbleless.
\end{lem}

\begin{proof}
Since $f'(k_t)=R_t<R_{t-1}=f'(k_{t-1})$ and $f''<0$, we obtain $k_t>k_{t-1}$. Since $R_t\le G$, by \eqref{eq:system_p} we obtain
\begin{equation*}
    p_t=\frac{R_t}{G}p_{t-1}-d_t\le p_{t-1}.
\end{equation*}
By Lemma \ref{lem:g} and \eqref{eq:system_k}, we obtain $k_{t+1}=g(k_t,p_t)>g(k_{t-1},p_{t-1})=k_t$ and hence $R_{t+1}=f'(k_{t+1})<f'(k_t)=R_t\le G$. By induction, $G\ge R_t>R_{t+1}>\dots>0$, so $\set{R_t}$ converges to some $R\in [0,G)$. Take $\epsilon>0$ such that $R+\epsilon<G$. Then \eqref{eq:system_p} implies
\begin{equation*}
    p_t\le \frac{R+\epsilon}{G}p_{t-1}
\end{equation*}
for large enough $t$, so $p_t\to 0$. The rest of the proof is the same as Lemma \ref{lem:Rinc}.
\end{proof}

\begin{lem}\label{lem:Rdec2}
Suppose Assumptions \ref{asmp:G}--\ref{asmp:D} hold. Assume $R_{t_0}<R_{t_0-1}$ for some $t_0\ge 1$, and $R_t<R_{t-1}\implies R_t>G$. Then $R_t>G$ for $t\ge t_0$ and exactly one of the following statements is true.
\begin{enumerate}
    \item\label{item:Rdec21} $\set{k_t,p_t,R_t}$ converges to $(k,0,R)$ satisfying $k\in \set{0}\cup \cK^*$ and $R=f'(k)\ge G$, and the equilibrium is bubbleless.
    \item\label{item:Rdec22} $\set{k_t,p_t,R_t}$ converges to $(k,p,G)$ satisfying $k=g(k,p)$, $p>0$, and $G=f'(k)$, and the equilibrium is asymptotically bubbly.
\end{enumerate}
\end{lem}

\begin{proof}
By assumption, $R_{t_0}>G$. Suppose $R_t>G$ for some $t\ge t_0$. If $R_{t+1}\ge R_t$, then clearly $R_{t+1}>G$. If $R_{t+1}<R_t$, by assumption $R_{t+1}>G$. By induction, $R_t>G$ for all $t\ge t_0$.

Since $R_t>G$ for all sufficiently large $t$, we have $v_t\to 0$ by \eqref{eq:vt_ub}. If the equilibrium is bubbleless, then $p_t=v_t\to 0$, and Lemma \ref{lem:kconverge} implies $k\to \set{0}\cup \cK^*$. Since $R_t>G$ eventually, we have $R_t=f'(k_t)\to R\ge G$.

If the equilibrium is bubbly, let $p_t=v_t+b_t$ with $b_t>0$. The no-arbitrage condition implies that the growth rate of the detrended bubble satisfies
\begin{equation}
    \frac{b_{t+1}}{b_t}=\frac{R_{t+1}}{G}\ge 1 \label{eq:b_growth}
\end{equation}
for all sufficiently large $t$, so $\set{b_t}$ is eventually increasing. By Lemma \ref{lem:k}, $b_t\le p_t\le f(k_t)$ is bounded, so $\set{b_t}$ converges to some $b>0$. Letting $t\to\infty$ in $p_t=v_t+b_t$ and \eqref{eq:b_growth}, we obtain $R_t\to G$ and $p_t\to b$. Therefore, the equilibrium is asymptotically bubbly.
\end{proof}

\begin{proof}[Proof of Proposition \ref{prop:longrun}]
All claims are immediate from Lemmas \ref{lem:Rinc}--\ref{lem:Rdec2} except $R\ge G_d$ in case \ref{item:lr_asymbubbleless}. Suppose to the contrary that $R<G_d$. Take $\epsilon>0$ such that $R+\epsilon<G_d-\epsilon$. Then we can take large enough $T>0$ such that $R_t\le R+\epsilon$ for all $t\ge T$ and a subsequence $T\le t_1<t_2<\dotsb$ such that $D_{t_n}\ge (G_d-\epsilon)^{t_n}$ for all $n$. Then we can bound the fundamental value \eqref{eq:Vt1} from below as
\begin{equation*}
    V_T\ge \sum_{n=1}^\infty (R+\epsilon)^{T-t_n}(G_d-\epsilon)^{t_n}=(R+\epsilon)^T\sum_{n=1}^\infty\left(\frac{G_d-\epsilon}{R+\epsilon}\right)^{t_n}=\infty,
\end{equation*}
which is a contradiction.
\end{proof}

\subsection{Proof of Theorem \ref{thm:eqset}}

The proof uses three ingredients. Proposition \ref{prop:p0} implies that $\cP_0$ is a compact interval. Lemma \ref{lem:kconverge} characterizes equilibria for which $p_t\to 0$, while Lemma \ref{lem:right} shows that any such equilibrium converging to a positive fixed point can be perturbed upward to produce a continuum of equilibria.

\begin{lem}\label{lem:right}
Suppose Assumptions \ref{asmp:G}--\ref{asmp:D} hold, $g(k,0)>k$ for small enough $k>0$, and $\sup_{k\in \cK^*}f'(k)<G$, where $\cK^*$ is as in \eqref{eq:cK}. If $\set{(k_t,p_t)}_{t=0}^\infty$ is an equilibrium satisfying $(k_t,p_t)\to (k,0)$ with $k\in \cK^*$, for $\eta>0$ small enough, the sequence $\set{(k_t^\eta,p_t^\eta)}_{t=0}^{\infty}$ defined by $p_0^\eta=p_0+\eta$ and \eqref{eq:system} is a bubbly but asymptotically bubbleless equilibrium.
\end{lem}

\begin{proof}
By Lemma \ref{lem:diamond}, $\cK^*$ is nonempty and compact, so $0<\ubar{k}\coloneqq \min \cK^*$ exists. By assumption, $f'(\ubar{k})<G$. Take any equilibrium $\set{(k_t,p_t)}_{t=0}^\infty$ satisfying $(k_{t},p_t)\to (k,0)$ with $k\in \cK^*$. Then clearly $k\ge \ubar{k}$ and $f'(k)\le f'(\ubar{k})<G$. Take $\epsilon>0$ such that $f'(\ubar{k}-\epsilon)<G$. Since $g(k,0)>k$ for $k\in (0,\ubar{k})$, by the definition of $p(k)$ in \eqref{eq:pk} and the subsequent remark, we have $p(\ubar{k}-\epsilon)>0$. 

Since $(k_t,p_t)\to (k,0)$, we can take $T>0$ such that $k_T>\ubar{k}-\epsilon$ and $p_T<p(\ubar{k}-\epsilon)$. By continuity, for sufficiently small $\eta>0$, the sequence $\set{(k_t^\eta,p_t^\eta)}_{t=0}^T$ defined by $p_0^\eta=p_0+\eta$ and \eqref{eq:system} is well defined and satisfies $0<k_t^\eta\le k_t$, $p_t<p_t^\eta$ for all $t=0,\dots,T$ and $k_T^\eta>\ubar{k}-\epsilon$, $p_T^\eta<p(\ubar{k}-\epsilon)$. Let us show by induction on $t$ that the sequence $\set{(k_t^\eta,p_t^\eta)}_{t=0}^\infty$ is well defined and $\ubar{k}-\epsilon<k_t^\eta\le k_t$, $p_t<p_t^\eta<p(\ubar{k}-\epsilon)$ for all $t\ge T$. If $t=T$, the claim is obvious. If the claim holds for some $t$, then by \eqref{eq:system_k}, Lemma \ref{lem:g}, and \eqref{eq:pk}, we obtain
\begin{align*}
    k_{t+1}^\eta&=g(k_t^\eta,p_t^\eta)>g(\ubar{k}-\epsilon,p(\ubar{k}-\epsilon))=\ubar{k}-\epsilon,\\
    k_{t+1}^\eta&=g(k_t^\eta,p_t^\eta)\le g(k_t,p_t)=k_{t+1}.
\end{align*}
Using \eqref{eq:system_p}, we obtain
\begin{align*}
    p_{t+1}^\eta&=\frac{f'(k_{t+1}^\eta)}{G}p_t^\eta-d_{t+1}\le \frac{f'(\ubar{k}-\epsilon)}{G}p_t^\eta\le p_t^\eta<p(\ubar{k}-\epsilon),\\
    p_{t+1}^\eta&=\frac{f'(k_{t+1}^\eta)}{G}p_t^\eta-d_{t+1}> \frac{f'(k_{t+1})}{G}p_t-d_{t+1}=p_{t+1},
\end{align*}
so the claim also holds for $t+1$. By induction, $\set{(k_t^\eta,p_t^\eta)}_{t=0}^\infty$ is well defined and hence it is an equilibrium. It is bubbly because $p_t^\eta>p_t$ implies $b_t^\eta>b_t\ge 0$ by Proposition \ref{prop:p0}\ref{item:p0_monotone}. Since the construction works for any sufficiently small $\eta>0$, take such $\eta<\eta'$. Because the equilibrium associated with $\eta'$ has a higher initial price, the equilibrium associated with $\eta<\eta'$ is not the maximum price equilibrium; Corollary \ref{cor:unique_bubble} implies that it is asymptotically bubbleless.
\end{proof}

\begin{proof}[Proof of Theorem \ref{thm:eqset}]
By Proposition \ref{prop:p0}\ref{item:p0_interval}, the equilibrium set $\cP_0$ is a nonempty compact interval.

Suppose first that $\cP_0$ has an empty interior. Then $\cP_0$ is a singleton, so the equilibrium is unique. If the equilibrium is asymptotically bubbly, noting that Proposition \ref{prop:longrun} covers all cases, statement \ref{item:eqset2} holds. If the equilibrium is not asymptotically bubbly, again by Proposition \ref{prop:longrun}, we must have $p_t\to 0$. By Lemma \ref{lem:kconverge}, we have $k_t\to k\in \set{0}\cup \cK^*$. If $k>0$, by Lemma \ref{lem:right}, there exists a continuum of equilibria, which is a contradiction. Therefore $k=0$ and $R=f'(0)=\infty$. By Lemma \ref{lem:impossible2}, the equilibrium is bubbleless and statement \ref{item:eqset1} holds.

Suppose next that $\cP_0$ has a nonempty interior and write $\cP_0=[\ubar{p}_0,\bar{p}_0]$, where $\ubar{p}_0<\bar{p}_0$.

\medskip
\noindent
\ref{item:eqset3a} By Proposition \ref{prop:p0}\ref{item:p0_monotone}, any $p_0>\ubar{p}_0$ is bubbly, so statement \ref{item:eqset3a} holds.

\medskip
\noindent
\ref{item:eqset3c} Suppose to the contrary that $p_0=\bar{p}_0$ is not asymptotically bubbly. By Proposition \ref{prop:p0}\ref{item:p0_monotone}, it must be bubbly but asymptotically bubbleless. Noting that Proposition \ref{prop:longrun} covers all cases, it must be $p_t\to 0$. By Lemma \ref{lem:kconverge}, we have $k_t\to k\in \set{0}\cup \cK^*$. If $k=0$, by Lemma \ref{lem:impossible2}, the equilibrium is unique, which contradicts the assumption that $\cP_0$ has a nonempty interior. Therefore $k>0$. By Lemma \ref{lem:right}, for sufficiently small $\eta>0$, $p_0^\eta=\bar{p}_0+\eta$ is also an equilibrium, which contradicts the maximality of $\bar{p}_0$. Therefore $\bar{p}_0$ must be asymptotically bubbly, and statement \ref{item:eqset3c} holds by Proposition \ref{prop:longrun}.

\medskip
\noindent
\ref{item:eqset3b} Suppose $p_0\in [\ubar{p}_0,\bar{p}_0)$. By Corollary \ref{cor:unique_bubble}, the equilibrium is not asymptotically bubbly. Noting that Proposition \ref{prop:longrun} covers all cases, it must be $p_t\to 0$. By Lemma \ref{lem:kconverge}, we have $k_t\to k\in \set{0}\cup \cK^*$. If $k=0$, by Lemma \ref{lem:impossible2}, the equilibrium is unique, which contradicts the assumption that $\cP_0$ has a nonempty interior. Therefore $k>0$, and statement \ref{item:eqset3b} holds.
\end{proof}

\subsection{Proof of Theorem \ref{thm:necessity}}

By Proposition \ref{prop:exist}, there exists an equilibrium. By Proposition \ref{prop:longrun}, $\set{(k_t,p_t,R_t)}$ converges to some $(k,p,R)$. By Lemma \ref{lem:diamond}\ref{item:diamond_ub}, $k_t\le k_t^*$ for all $t$, and hence $k\le k^*$. Proposition \ref{prop:longrun} \ref{item:lr_asymbubbleless} cannot occur, because it would imply $k\in \cK^*\cap (0,k^*]$ and $f'(k)\ge G_d$, contradicting \eqref{eq:necessity}. Suppose Proposition \ref{prop:longrun}\ref{item:lr_bubbleless} occurs. If $k>0$, then again $k\in \cK^*\cap (0,k^*]$, so $G\le f'(k)<G_d\le G$, a contradiction. Therefore $k=0$ and statement \ref{item:eqset1} of Theorem \ref{thm:eqset} holds. If Proposition \ref{prop:longrun}\ref{item:lr_bubbly} occurs, then statement \ref{item:eqset2} of Theorem \ref{thm:eqset} obviously holds. The uniqueness of equilibrium follows from Lemma A.4 of \citet{PhamToda2026ECMA}.

Finally, we show the existence of the threshold $\kappa\in [0,\infty]$ if \eqref{eq:necessity} holds for all $k\in \cK^*$. Define the set
\begin{equation*}
    \cK_0 \coloneqq \set{k_0>0:\text{Case \ref{item:eqset2} in Theorem \ref{thm:eqset} occurs}}
\end{equation*}
and let $\kappa=\inf \cK_0$ (so by convention $\kappa=\infty$ if $\cK_0=\emptyset$). By definition, if $k_0<\kappa$, then $k_0\notin \cK_0$, so case \ref{item:eqset1} occurs. If $k_0>\kappa$, we can take some $k_0'\in [\kappa,k_0]$ such that $k_0'\in \cK_0$. Lemma A.5 of
\citet{PhamToda2026ECMA} then implies that $k_0\in \cK_0$. \hfill \qedsymbol

\subsection{Proof of Theorem \ref{thm:curse}}

We first record the behavior of $f'(k)$ near $k=0$. Since $\omega'(k)=-kf''(k)$, for any fixed $\delta>0$, integration by parts gives
\begin{align*}
    f'(k)-f'(\delta)=-\int_k^\delta f''(x)\diff x=\int_k^\delta\frac{\omega'(x)}{x}\diff x=\frac{\omega(\delta)}{\delta}-\frac{\omega(k)}{k}+\int_k^\delta \frac{\omega(x)}{x^2}\diff x.
\end{align*}
Using \eqref{eq:w_lim}, we obtain
\begin{equation*}
\int_k^\delta \frac{\omega(x)}{x^2}\diff x\sim
\begin{cases*}
\frac{L}{1-\alpha}k^{\alpha-1}
    & if $\alpha\in(0,1)$,\\
-L\log k & if $\alpha=1$
\end{cases*}
\end{equation*}
as $k\to 0$. Consequently,
\begin{equation}
f'(k)\sim
\begin{cases*}
\frac{\alpha L}{1-\alpha}k^{\alpha-1}
    & if $\alpha\in(0,1)$,\\
-L\log k & if $\alpha=1$.
\end{cases*}\label{eq:f'k}
\end{equation}
For any equilibrium, feasibility and no-arbitrage imply
\begin{align*}
Gk_{t+1}+p_t&=s_t\le \omega(k_t), \\
p_t&=\frac{G}{f'(k_{t+1})}(p_{t+1}+d_{t+1})\ge \frac{Gd_{t+1}}{f'(k_{t+1})}
\end{align*}
and hence
\begin{equation}
Gk_{t+1}+\frac{Gd_{t+1}}{f'(k_{t+1})}
\le\omega(k_t). \label{eq:kineq}
\end{equation}
For each $t$, the left-hand side of \eqref{eq:kineq}, viewed as a function of $k_{t+1}$, is strictly increasing. Consequently, if a continuous
function $\phi$ with $\phi(0)=0$ and an integer $n$ satisfy
\begin{equation}
\omega(\phi(d_{t+n}))
\le
G\phi(d_{t+n+1})
+\frac{Gd_{t+1}}{f'(\phi(d_{t+n+1}))} \label{eq:phi}
\end{equation}
for all $t$, then $k_0\le\phi(d_n)$ implies by induction that $k_t\le\phi(d_{t+n})$ for all $t$.

We now construct $\phi$. If $\alpha<1$, set $\eta=1-\alpha$. If $\alpha=1$, choose $\eta\in (0,\log r_2/\log r_1)$, where we note that $\log r_j<0$. In either case,
\begin{equation}
B\coloneqq \eta\log r_1-\alpha\log r_2>0. \label{eq:B}
\end{equation}
For $\alpha<1$, \eqref{eq:B} follows from \eqref{eq:alpha};
for $\alpha=1$, it follows from the choice of $\eta$.

By \eqref{eq:f'k}, there exist $C_1,C_2>0$ such that, for all sufficiently
small $k>0$,
\begin{equation}
f'(k)\le C_1k^{-\eta} \quad \text{and} \quad \omega(k)\le C_2k^\alpha,
\label{eq:f'omega}
\end{equation}
where the first inequality for $\alpha=1$ follows from
$-\log k=o(k^{-\eta})$. Choose $\gamma>0$ so large that $\lambda\coloneqq \log r_1+\gamma B>0$ and set $\phi(x)=x^\gamma$. For all sufficiently large $n$ and all
$t\ge 0$, \eqref{eq:f'omega} implies
\begin{equation}
\frac{
 f'((D_1r_1^{t+n+1})^\gamma)
 \omega((D_2r_2^{t+n})^\gamma)}
 {GD_1r_1^{t+1}} \le Cr_1^{-t-1-\gamma\eta(t+n+1)}r_2^{\alpha\gamma(t+n)}=C'\e^{-\lambda t-\gamma Bn} \label{eq:frac_ub}
\end{equation}
for some constants $C,C'>0$ independent of $t$ and $n$. Taking $n$
still larger makes the right-hand side of \eqref{eq:frac_ub} bounded above by 1. Hence
\begin{equation*}
\omega((D_2r_2^{t+n})^\gamma)
\le
\frac{GD_1r_1^{t+1}}
 {f'((D_1r_1^{t+n+1})^\gamma)}.
\end{equation*}
Using $D_1r_1^t\le d_t\le D_2r_2^t$ and the monotonicity of $\omega, 1/f'$, we obtain
\begin{equation*}
\omega(\phi(d_{t+n}))
\le
\frac{Gd_{t+1}}{f'(\phi(d_{t+n+1}))},
\end{equation*}
which implies \eqref{eq:phi}. Setting $\kappa=\phi(d_n)>0$, the comparison
argument gives $0<k_t\le\phi(d_{t+n})\to 0$ as $t\to\infty$. \hfill \qedsymbol

\subsection{Proof of Theorem \ref{thm:continuum}}
By Assumption \ref{asmp:f}, $k_g>0$ exists. By Lemma \ref{lem:diamond} and the capital over-accumulation condition \eqref{eq:over-accumulation}, we have $k_g<\ubar{k}\coloneqq \min \cK^*$. Since $g(k,0)>k$ for small enough $k$, by the definition of $\cK^*$, we have $g(k,0)>k$ for $k\in (0,\ubar{k})$. In particular, $g(k_g,0)>k_g$ and $p(k_g)>0$.

Since $V_0^*\le p(k_g)<\infty$, by Theorem \ref{thm:bubbleless}, there exists a unique bubbleless equilibrium, which we denote by $\set{(k_t,p_t)}_{t=0}^\infty$. By Corollary \ref{cor:unique_bubbleless}, it corresponds to $p_0=\min \cP_0$. By Lemma \ref{lem:diamond}\ref{item:diamond_ub}, we have $R_t\ge R_t^*$ and hence $p_0=v_0\le V_0^*\le p(k_g)$. Let us prove by induction that $k_t\ge k_g$ and $p_t\le p(k_g)$ for all $t$. The claim is trivial for $t=0$. If the claim holds for some $t$, the monotonicity of $g$ and the definition of $p$ imply $k_{t+1}=g(k_t,p_t)\ge g(k_g,p(k_g))=k_g$. Since $f'$ is decreasing and $f'(k_g)=G$, we have
\begin{equation*}
    p_{t+1}=\frac{f'(k_{t+1})}{G}p_t-d_{t+1}\le \frac{f'(k_g)}{G}p_t= p_t\le p(k_g).
\end{equation*}
Hence the claim holds for $t+1$. Since $k_t\ge k_g>0$, we cannot have $k_t\to 0$, so case \ref{item:eqset1} of Theorem \ref{thm:eqset} cannot occur. Since a bubbleless equilibrium exists, case \ref{item:eqset2} of Theorem \ref{thm:eqset} cannot occur. Therefore case \ref{item:eqset3} must occur. \hfill \qedsymbol

\subsection{Proof of Lemma \ref{lem:CD}}

\ref{item:lem_CD1} Let $x_t\coloneqq A\alpha k_t^\alpha/(Gk_{t+1})>0$. By definition, \eqref{eq:xdef_k} holds. Solving \eqref{eq:systemCD_k} for $p_t$ and using the definition of $\rho$ in \eqref{eq:rho}, we obtain
\begin{equation*}
    p_t=\beta A(1-\alpha)k_t^\alpha-Gk_{t+1}=\frac{A\alpha}{\rho}k_t^\alpha-Gk_{t+1},
\end{equation*}
which is \eqref{eq:xdef_p}. Changing $t$ to $t+1$ in \eqref{eq:systemCD_p} for $d_{t+1}$, we obtain \eqref{eq:xdef_d}. By assumption, $p_t>0$ and $d_t\ge 0$. Using \eqref{eq:xdef_k} and \eqref{eq:xdef_p}, we obtain
\begin{equation}
    0<p_t=\frac{A\alpha}{\rho}k_t^\alpha-\frac{A\alpha}{x_t}k_t^\alpha=A\alpha k_t^\alpha\left(\frac{1}{\rho}-\frac{1}{x_t}\right), \label{eq:ptx}
\end{equation}
which is equivalent to $x_t>\rho$ or \eqref{eq:xineq1}. Similarly, using \eqref{eq:xdef} and \eqref{eq:ptx}, we obtain
\begin{align}
    0\le d_{t+1}&=\frac{A\alpha}{G}k_{t+1}^{\alpha-1}A\alpha k_t^\alpha\left(\frac{1}{\rho}-\frac{1}{x_t}\right)-A\alpha k_{t+1}^\alpha\left(\frac{1}{\rho}-\frac{1}{x_{t+1}}\right) \notag \\
    &=A\alpha k_{t+1}^\alpha \left[\frac{A\alpha k_t^\alpha}{Gk_{t+1}}\left(\frac{1}{\rho}-\frac{1}{x_t}\right)-\left(\frac{1}{\rho}-\frac{1}{x_{t+1}}\right)\right] \notag \\
    &=A\alpha k_{t+1}^\alpha\left[\frac{x_t}{\rho}-1-\frac{1}{\rho}+\frac{1}{x_{t+1}}\right], \label{eq:kdx_d2}
\end{align}
which is equivalent to \eqref{eq:xineq2}.

\medskip
\noindent
\ref{item:lem_CD2} The proof that $\set{(k_t,p_t)}_{t=0}^\infty$ is an equilibrium is immediate by going in the reverse direction of \ref{item:lem_CD1}. \eqref{eq:kdx_k} is immediate from \eqref{eq:xdef_k} and induction. \eqref{eq:kdx_d} is proved in \eqref{eq:kdx_d2}. To show \eqref{eq:kdx_dp}, note that
\begin{align*}
    \frac{d_t}{p_t}&=\frac{A\alpha}{G}k_t^{\alpha-1}\frac{p_{t-1}}{p_t}-1 && (\because \eqref{eq:xdef_d})\\
    &=\frac{A\alpha}{G}k_t^{\alpha-1}\frac{A\alpha k_{t-1}^\alpha (1/\rho-1/x_{t-1})}{A\alpha k_t^\alpha (1/\rho-1/x_t)}-1 && (\because \eqref{eq:xdef_p})\\
    &=\frac{A\alpha k_{t-1}^\alpha}{Gk_t}\frac{1/\rho-1/x_{t-1}}{1/\rho-1/x_t}-1\\
    &=x_{t-1}\frac{1/\rho-1/x_{t-1}}{1/\rho-1/x_t}-1, && (\because \eqref{eq:xdef_k})
\end{align*}
which simplifies to \eqref{eq:kdx_dp}. \hfill \qedsymbol

\subsection{Proof of Lemma \ref{lem:kx}}

To simplify the notation, let $y_t\coloneqq \frac{A\alpha}{Gx_t}$ and $y=\frac{A\alpha}{Gx}\in [0,\infty)$.

Since $y_t\to y$, for any $\bar{y}>y$, we can take $T>0$ such that $y_t\le \bar{y}$ for $t\ge T$. Define the sequence $\set{\bar{k}_t}_{t=0}^\infty$ by $\bar{k}_t=k_t$ for $t\le T$ and $\bar{k}_{t+1}=\bar{y}\bar{k}_t^\alpha$ for $t\ge T$. Taking the logarithm and solving the linear difference equation in $\log \bar{k}_t$, we obtain
\begin{equation*}
    \log \bar{k}_t=\alpha^{t-T}\log \bar{k}_T+\frac{1-\alpha^{t-T}}{1-\alpha}\log \bar{y}
\end{equation*}
for $t\ge T$, so $\bar{k}_t\to \bar{y}^\frac{1}{1-\alpha}$. By monotonicity, $k_t\le \bar{k}_t$ holds for all $t$. Therefore,
\begin{equation}
    \limsup_{t\to\infty}k_t\le \limsup_{t\to\infty}\bar{k}_t=\bar{y}^\frac{1}{1-\alpha}. \label{eq:liminfk_ub}
\end{equation}
Sending $\bar{y}\downarrow y$, we obtain $\limsup_{t\to\infty}k_t\le y^\frac{1}{1-\alpha}$.  If $x=\infty$, then $y=0$, so we obtain $\lim_{t\to\infty}k_t=0=y^\frac{1}{1-\alpha}$. If $x<\infty$, then $y\in (0,\infty)$. An analogous argument using the lower bound yields $\liminf_{t\to\infty}k_t\ge y^\frac{1}{1-\alpha}$ and hence $\lim_{t\to\infty}k_t=y^\frac{1}{1-\alpha}$. \hfill \qedsymbol



\printbibliography

\newpage

\begin{center}
{\LARGE\bf Online Appendix (Not for publication)}
\end{center}

\begin{refsection}
\section{Detailed discussion of the literature}\label{sec:disc}

In this appendix, we discuss how our results relate to the analysis of \citet{Tirole1985}, \citet*{BosiHa-HuyLeVanPhamPham2018}, and \citet[\S V.A]{HiranoToda2025JPE}.

\subsection{Analysis of \texorpdfstring{\citet{Tirole1985}}{}}\label{subsec:disc_Tirole}

\citet{Tirole1985} denotes population by $(1+n)^t$ (where $n>0$ is the net population growth rate), so Assumption \ref{asmp:G} holds with $G>1$. Regarding the production function, \citet{Tirole1985} uses the notation $F(K,L)$ for output \emph{excluding} undepreciated capital and implicitly assumes Assumption \ref{asmp:f} and no capital depreciation, so our $F(K,L)$ corresponds to \citet{Tirole1985}'s $F(K,L)+K$. \citet{Tirole1985} also assumes constant rents (dividends), so Assumption \ref{asmp:D} holds with $G_d=1<1+n=G$. Thus, with respect to Assumptions \ref{asmp:G}, \ref{asmp:f}, \ref{asmp:D}, our setting is more general than \citet{Tirole1985}. To facilitate comparison, Table \ref{t:notation} compares our notation to his.

\begin{table}[!htb]
\centering
    \begin{tabular}{lcc}
    \toprule
    Description & \citet{Tirole1985} & Our paper \\
    \midrule
    Asset price & $(1+n)^ta_t$ & $P_t=p_tG^t$ \\
    Bubble & $(1+n)^tb_t$ & $B_t=b_tG^t$ \\
    Dividend & $R$ (constant) & $D_t=d_tG^t$ \\
    Fundamental value & $F_t=(1+n)^tf_t$ & $V_t=v_tG^t$ \\
    Gross population growth & $1+n>1$ & $G>0$ \\
    Gross return & $1+r_t$ & $R_t$ \\
    Long-run dividend growth & $1$ & $G_d=\limsup_{t\to\infty}D_t^{1/t}$\\
    Production function & $F(K,L)+K$ & $F(K,L)$ \\
    Utility function & $u(c^y,c^o)$ & $U(c_t^y,c_{t+1}^o)$ \\
    Young population & $(1+n)^t$ & $L_t=G^t$ \\
    \bottomrule
    \end{tabular}
    \caption{Notation in \citet{Tirole1985} and our paper.}
    \label{t:notation}
\end{table}

By Lemma \ref{lem:k}, rent and wage are related through capital as $R=f'(k)$ and $w=f(k)-kf'(k)$. Since $f''<0$, we may invert $f'$ and obtain $k=(f')^{-1}(R)$. \citet[Equation (4)]{Tirole1985} writes
\begin{equation}
    w=f(k)-kf'(k)\eqqcolon \phi(R) \label{eq:phi_Tirole}
\end{equation}
for $k=(f')^{-1}(R)$. Applying the chain rule, we obtain
\begin{equation*}
    \phi'(R)=-kf''(k)\frac{\diff k}{\diff R}=-kf''(k)\frac{1}{f''(k)}=-k<0.
\end{equation*}

Let $s(w,R)$ be the optimal savings function given the wage $w$ and the gross risk-free rate $R$. Then the equilibrium condition \eqref{eq:eqcond} is equivalent to
\begin{equation}
    G(f')^{-1}(R_{t+1})+p_t=s(w_t,R_{t+1}). \label{eq:eqcond_Tirole}
\end{equation}
\citet[Equation (7)]{Tirole1985} imposes the following high-level assumption: the equilibrium condition \eqref{eq:eqcond_Tirole} can be uniquely solved as
\begin{equation}
    R_{t+1}=\psi(w_t,p_t), \label{eq:psi}
\end{equation}
where $\psi_w<0$ and $\psi_p>0$. However, \citet{Tirole1985} does not provide any conditions on exogenous objects that guarantee this monotonicity condition. We may justify this assumption as follows. Noting that the equilibrium conditions \eqref{eq:eqcond}, \eqref{eq:eqcond_Tirole} are equivalent, it follows that our function $g$ in Lemma \ref{lem:g} satisfies
\begin{equation*}
    g(k,p)=(f')^{-1}(\psi(f(k)-kf'(k),p)).
\end{equation*}
Noting that $k\mapsto f(k)-kf'(k)$ is strictly increasing, the existence and monotonicity of $\psi$ is equivalent to those of $g$. Thus, we see that the high-level assumption of \citet{Tirole1985} is justified by our Assumption \ref{asmp:s}, which is satisfied under the conditions on exogenous objects in Lemma \ref{lem:s}.

Letting $\phi,\psi$ be defined as \eqref{eq:phi_Tirole} and \eqref{eq:psi}, the discussion around Equation (8) of \citet{Tirole1985} assumes that there exists a unique $R>0$ such that $R=\psi(\phi(R),0)$. Furthermore, if $R<G$, there exists a unique $b>0$ such that $G=\psi(\phi(G),b)$. In our notation, the existence and uniqueness of such $R$ is equivalent to saying that the set of steady-state capital-labor ratio $\cK^*=\set{k>0:k=g(k,0)}$ in \eqref{eq:cK} is a singleton, whereas we impose no condition on $\cK^*$: see Remark \ref{rem:singlecross}. Once we assume the existence and uniqueness of $R$, the existence and uniqueness of $b>0$ above is immediate due to the monotonicity of $g$.

In summary, with regard to assumptions, our setting is strictly more general than \citet{Tirole1985} and justifies his high-level assumptions.

We now turn to the discussion of results. We first note that \citet{Tirole1985} does not prove the existence of equilibrium in a general setting, whereas we provide it in Proposition \ref{prop:exist}. The main result of \citet{Tirole1985} is his Proposition 1. We quote the essential parts except that we modify the notation according to Table \ref{t:notation}. Following his assumptions, we assume that $\cK^*=\set{k}$ is a singleton and $R=f'(k)$.

\begin{oneshot}[Proposition 1 of \citet{Tirole1985}]
\quad
\begin{enumerate}[(a)]
    \item If $R>G$, there exists a unique equilibrium. This equilibrium is bubbleless and the interest rate converges to $R$.
    \item If $G_d=1<R<G$, there exists a maximum feasible bubble $\hat{b}_0>0$, such that: (i) for any $b_0\in [0,\hat{b}_0)$, there exists a unique equilibrium with initial bubble $b_0$. This equilibrium is asymptotically bubbleless and the interest rate converges to $R$.
    (ii) There exists a unique equilibrium with initial bubble $\hat{b}_0$. The bubble per capita converges to $b$ and the interest rate converges to $G$.
    \item If $R<1=G_d$, there exists no bubbleless equilibrium. There exists a unique bubbly equilibrium. It is asymptotically bubbly and the interest rate converges to $G$.
\end{enumerate}
\end{oneshot}

As discussed in Remark \ref{rem:collapse}, each case of Proposition 1 of \citet{Tirole1985} is incorrect as stated. \citet{Tirole1985} proves Proposition 1 based on Lemmas 1--10 in his paper. Lemma 1 claims that under the condition $R>G_d$ in our notation, there exists a unique bubbleless equilibrium and that $\lim_{t\to\infty}R_t=R$. We find this lemma problematic for a few reasons.

First, the continuity argument in \citet{Tirole1985} is loose. He implicitly assumes that the present value of dividends computed with the Diamond bubbleless and rentless interest rates is finite. This condition is exactly \eqref{eq:bubbleless_cond}. However, he did not prove that his function $\Gamma$ is continuous. Our proof of the existence of bubbleless equilibrium (Theorem \ref{thm:bubbleless}) is different and directly implies existence in \citet{Tirole1985}'s setting.

Second, \citet{Tirole1985}'s proof of $\lim_{t\to\infty}R_t=R$ is incomplete. To see why, it is useful to consider the following exhaustive and mutually exclusive cases:
\begin{enumerate}
    \item\label{item:Rinc} $R_t\ge R_{t-1}$ for all $t$,
    \item\label{item:Rdec1} $R_t<R_{t-1}$ and $R_t\le G$ for some $t$,
    \item\label{item:Rdec2} $R_t<R_{t-1}$ for some $t$, and $R_t>G$ for any such $t$.
\end{enumerate}
Note that the cases \ref{item:Rinc}--\ref{item:Rdec2} parallel the assumptions in our Lemmas \ref{lem:Rinc}--\ref{lem:Rdec2}. In the ``convergence'' proof of Lemma 1, \citet[p.~1522]{Tirole1985} only considers case \ref{item:Rdec1}. However, in case \ref{item:Rinc} we cannot exclude the possibility of $\lim_{t\to\infty}R_t=\infty$, and in case \ref{item:Rdec2} we can generally only conclude that $\liminf_{t\to\infty}R_t\ge G$. (See the proof of Lemma \ref{lem:Rdec2}.) In fact, our Example \ref{exmp:k0} provides a counterexample in which $\lim_{t\to\infty}R_t=\infty$.

Lemma 2 of \citet{Tirole1985} corresponds to our Lemma \ref{lem:Rdec1}. The first and second parts of the proof of his Lemma 3 correspond to our Lemmas \ref{lem:Rinc} and \ref{lem:Rdec2}. His Lemmas 4, 6, and 10 correspond to our Proposition \ref{prop:p0}. Here he proves that the equilibrium set is an interval and equilibria satisfy monotonicity. His Lemma 5 corresponds to our Corollary \ref{cor:unique_bubble}. In all these results, we follow the same proof strategy as \citet{Tirole1985} and hence we do not claim any originality.

Lemma 7 of \citet{Tirole1985} claims that if $R<G$, then we can construct a bubbly equilibrium if the initial bubble is sufficiently small. However, the proof depends on the convergence result in his Lemma 1, which is incorrect. We construct a continuum of bubbly equilibria in Theorem \ref{thm:continuum} using a different approach. Lemma 9 of \citet{Tirole1985} also assumes convergence and is problematic.

Lemma 8 of \citet{Tirole1985} shows that there exists no bubbly equilibrium if $R>G$. Lemma 2.3 of \citet{PhamToda2026ECMA} relaxes the assumptions and also establishes the uniqueness of equilibrium.

In summary, the main result of \citet{Tirole1985}, his Proposition 1, requires the following qualifications.
\begin{enumerate}[(a)]
    \item Regarding Proposition 1(a), although the existence and uniqueness of equilibrium (which is bubbleless) follows from Lemma 2.3 of \citet{PhamToda2026ECMA}, we cannot conclude that $\set{R_t}$ converges to $R<\infty$. By Theorem \ref{thm:curse}, $R_t\to\infty$ and hence $k_t\to 0$ is possible.
    \item Regarding Proposition 1(b), the proof is incomplete as it depends on the problematic Lemmas 1, 7, and 9. By Theorem \ref{thm:continuum}, we know that the equilibrium set $\cP_0=[\ubar{p}_0,\bar{p}_0]$ is a nondegenerate compact interval and any $p_0\in (\ubar{p}_0,\bar{p}_0)$ is bubbly but asymptotically bubbleless, which requires the condition $V_0^*\le p(k_g)$.
    \item Regarding Proposition 1(c), although the existence and uniqueness of equilibrium follows from Theorem \ref{thm:necessity}, we cannot conclude that the equilibrium is asymptotically bubbly. Indeed, Proposition 1 of \citet{PhamToda2026ECMA} provides a counterexample in which there exists a unique bubbleless equilibrium with $k_t\to 0$.
\end{enumerate}

\subsection{Analysis of \texorpdfstring{\citet*{BosiHa-HuyLeVanPhamPham2018}}{}}\label{subsec:disc_Bosi}

\citet*{BosiHa-HuyLeVanPhamPham2018} consider a model like ours but introduce forward (or descending) altruism. If we remove altruism, the model in \citet*{BosiHa-HuyLeVanPhamPham2018} reduces to ours. They prove that there is no bubble if $\sum_{t=1}^\infty d_t=\infty$ (Corollary 2) or $f'(k^*)>G$ (Proposition 2.1), which correspond to our Lemmas \ref{lem:impossible1}, \ref{lem:impossible2}. Proposition 2 of \citet*{BosiHa-HuyLeVanPhamPham2018} shows that when $f'(k^*)<G$ and $\set{d_t}$ is decreasing and converges to zero, then any equilibrium must be in one of three cases:
\begin{enumerate}
\item\label{item:Bosi1} $\liminf_{t\to \infty}k_t<k$, where $k$ satisfies $f'(k)=G$. In this case, the
equilibrium solution is bubbleless and unique.
\item\label{item:Bosi2} $\lim_{t\to \infty }k_t=k^*$ and $\lim_{t\to \infty}p_{t}=0$, where $k^*=g(k^*,0)$.
\item\label{item:Bosi3} $\lim_{t\to \infty}k_{t}=k$ and $\lim_{t\to
\infty}p_t=p$, where $f'(k)=G$ and $k=g(k,p)$.
\end{enumerate}
This result corresponds to our Proposition \ref{prop:longrun} but there are two differences. First, concerning the set of fixed points, Assumptions 5 and 6 (single-crossing conditions) of \citet*{BosiHa-HuyLeVanPhamPham2018} are relatively high-level, whereas our assumptions are explicit. Second, \citet{BosiHa-HuyLeVanPhamPham2018} require $\set{d_t}$ to be decreasing and to converge to zero, whereas Proposition \ref{prop:longrun} requires $\sum d_t<\infty$. The assumptions are not nested: summability does not require monotonicity, while monotone convergence to zero does not imply summability.

\citet*{BosiHa-HuyLeVanPhamPham2018} also consider a specific model with Cobb-Douglas production function and logarithmic utility and prove in Proposition 3 that when $f'(k^*)<G$ and $\lim_{t\to\infty} d_t=0$, case \ref{item:Bosi1} above must satisfy $\liminf_{t\to\infty}k_t=0$. They provide Example 1 where $k_t,p_t,d_t$ all converge to zero and Example 2 where there exists a continuum of bubbly equilibria.

Our paper provides more results. With respect to \citet*{BosiHa-HuyLeVanPhamPham2018}, our main new results are
\begin{enumerate*}
    \item the existence of bubbleless equilibrium (Theorem \ref{thm:bubbleless}),
    \item  showing all possible forms of the equilibrium set $\cP_0$ (Theorem \ref{thm:eqset}),
    \item a sufficient condition to rule out all bubbly but asymptotically bubbleless equilibria (Theorem  \ref{thm:necessity}),
    \item a general condition for capital collapse (Theorem \ref{thm:curse}), and
    \item a general condition for the existence of a continuum of equilibria (Theorem \ref{thm:continuum}).
\end{enumerate*}
By developing their approach, we provide more examples with numerical simulations.

\subsection{Analysis of \texorpdfstring{\citet{HiranoToda2025JPE}}{}}\label{subsec:disc_HT}

\citet{HiranoToda2025JPE} establish the necessity of bubbles (\ie, asset price bubble emerges in all equilibria) under some conditions in modern macro-finance models. Their main result can be roughly stated as follows. Let $G>0$ be the long-run economic growth rate, $G_d$ in \eqref{eq:Gd} the long-run dividend growth rate, and $R\ge 0$ the bubbleless interest rate (the interest rate that prevails in the absence of the long-lived asset). If the bubble necessity condition
\begin{equation}
    R<G_d<G \label{eq:necessity_HT}
\end{equation}
holds, then all equilibria are asymptotically bubbly in the sense that $P_t>V_t$ and $\liminf_{t\to\infty}P_t/G^t>0$.

The main analysis of \citet{HiranoToda2025JPE} concerns an OLG endowment economy. However, they also consider infinite-horizon or production economies. In particular, \citet[\S V.A]{HiranoToda2025JPE} consider a particular application to \citet{Tirole1985}'s model with log utility.

\begin{oneshot}[Theorem 3 of \citet{HiranoToda2025JPE}]
\itshape
Consider the model in \S\ref{subsec:model_tirole} with log utility \eqref{eq:utility_log}. Suppose Assumptions \ref{asmp:G}, \ref{asmp:f} hold with $G=1$. If there exists $k^*>0$ such that
\begin{equation}
    \beta F_L(k,1)-k\begin{cases*}
        >0 & if $0<k<k^*$,\\
        =0 & if $k=k^*$,\\
        <0 & if $k>k^*$
    \end{cases*} \label{eq:K*}
\end{equation}
and
\begin{equation}
	R\coloneqq F_K(k^*,1)<G_d\coloneqq \limsup_{t\to\infty}D_t^{1/t}<1\eqqcolon G, \label{eq:necessity_HT_diamond}
\end{equation}
then any equilibrium with $\liminf_{t\to\infty}K_t>0$ is asymptotically bubbly.
\end{oneshot}

Our Theorem \ref{thm:necessity} is strictly stronger than \citet[Theorem 3]{HiranoToda2025JPE}, henceforth HT3. Regarding the assumption, HT3 assumes log utility, which satisfies Assumption \ref{asmp:s} by Lemma \ref{lem:s}. Then $\beta F_L(k,1)-k=g(k,0)-k$ in our notation. Therefore, condition \eqref{eq:K*} implies that the set $\cK^*$ in \eqref{eq:cK} is a singleton. Under this condition, \eqref{eq:necessity_HT_diamond} is equivalent to \eqref{eq:necessity}. Therefore, the assumptions in Theorem \ref{thm:necessity} are weaker. Regarding the conclusions, all HT3 shows is that equilibria satisfying $\liminf_{t\to\infty}K_t>0$ are asymptotically bubbly. In contrast, Theorem \ref{thm:necessity} shows the existence and uniqueness of equilibrium and convergence results.

\section{Minimum equilibrium price need not be bubbleless}
\label{sec:bubbly-minimum}

Theorem~\ref{thm:eqset}\ref{item:eqset3b} states that every equilibrium below the upper endpoint is asymptotically bubbleless, but it does not claim that the lower endpoint is bubbleless. The following proposition shows that the minimum can indeed be bubbly.

\begin{prop}[Bubbly lower endpoint]
\label{prop:bubbly-minimum}
There exists an economy satisfying all the assumptions of Theorem \ref{thm:eqset} for which case \ref{item:eqset3} occurs and the equilibrium associated with $\ubar{p}_0=\min \cP_0$ is bubbly.
\end{prop}

\begin{proof}
We construct the technology so that the Diamond dynamics have three isolated
steady states.  The middle steady state is locally unstable, while the largest
steady state has an interest rate below the long-run dividend growth rate.
These two properties will force the equilibrium constructed below to be the
lower endpoint of the equilibrium set. The proof is by reverse-engineering, as in \S\ref{sec:example}.

\paragraph{Step 1: primitives.}

Let $G=1$ and consider the log utility \eqref{eq:utility_log} with $\beta=19/20$. Optimal savings is $s(w,R)=\beta w$, so Assumption \ref{asmp:s} holds. Let
\begin{equation*}
    (\kappa_1,\kappa_2,\kappa_3)=(9/10,1,13/10)
\end{equation*}
and define $m:[0,\infty)\to \R$ by
\begin{equation}
\label{eq:bubbly-min-m}
 m(k)=
 \begin{cases*}
 k+\frac{1}{45}k(\kappa_1-k)
 & for $k\in [0,\kappa_1)$,\\
 k+\frac{1}{2}(\kappa_1-k)(\kappa_2-k)(\kappa_3-k)
 & for $k\in [\kappa_1,\kappa_3]$,\\
 \kappa_3+\frac{47}{1000}
 \left(1-\exp\left(-20(k-\kappa_3)\right)\right)
 & for $k\in (\kappa_3,\infty)$.
 \end{cases*}
\end{equation}

Then $m$ is continuously differentiable, $m(k)>0$ for $k>0$, and $m'(k)>0$. To see this, we first show $m'(k)>0$. This is obvious for $k\in [0,\kappa_1)\cup (\kappa_3,\infty)$; and for $k\in [\kappa_1,\kappa_3]$, $m'(k)$ is a concave quadratic function and
\begin{align*}
    m'(\kappa_1)&=1-\frac{1}{2}(\kappa_2-\kappa_1)(\kappa_3-\kappa_1)>0,\\
    m'(\kappa_3)&=1-\frac{1}{2}(\kappa_3-\kappa_1)(\kappa_3-\kappa_2)>0
\end{align*}
because $\abs{\kappa_i-\kappa_j}<1$ for all $i,j$ pairs. Concavity of $m'(k)$ implies $m'(k)>0$ for $k\in [\kappa_1,\kappa_3]$. Furthermore, the function values and derivatives match at $k=\kappa_1,\kappa_3$ by direct computation, so $m$ is continuously differentiable and $m'>0$. Since $m$ is strictly increasing, $m(k)>m(0)=0$ for $k>0$. Moreover,
\begin{equation}
\label{eq:bubbly-min-sign}
 m(k)-k
 \begin{cases*}
 =0 & if $k\in \set{0,\kappa_1,\kappa_2,\kappa_3}$,\\
 >0 & if $k\in (0,\kappa_1)\cup (\kappa_2,\kappa_3)$,\\
 <0 & if $k\in (\kappa_1,\kappa_2)\cup (\kappa_3,\infty)$.
 \end{cases*}
\end{equation}
For $k\le \kappa_3$, \eqref{eq:bubbly-min-sign} is obvious because $m(k)-k$ is a polynomial with roots in $\set{0,\kappa_1,\kappa_2,\kappa_3}$; for $k>\kappa_3$, putting $y=k-\kappa_3>0$, we have
\begin{equation*}
    m(k)-k=-y+\frac{47}{1000}(1-\e^{-20y})<-y+\frac{47}{1000}(20y)=-\frac{3}{50}y<0.
\end{equation*}

Set $\bar{R}=4/5$ and define
\begin{align}
\label{eq:bubbly-min-Rf}
 R(k)&\coloneqq \bar{R}+\int_k^{\kappa_2}\frac{m'(x)}{\beta x}\diff x, &
 f(k)&\coloneqq \frac{m(k)}{\beta}+kR(k),
\end{align}
and $F(K,L)\coloneqq Lf(K/L)$. For $k\in (0,\kappa_1]$, by direct integration we have
\begin{equation*}
    R(k)=R(\kappa_1)+\frac{102}{95}\log\frac{\kappa_1}{k}+\frac{8}{171}(k-\kappa_1).
\end{equation*}
Hence $kR(k)\to 0$ as $k\downarrow 0$, so by \eqref{eq:bubbly-min-Rf} we have $f(0)=0$ and
\begin{align}
\label{eq:bubbly-min-f-properties}
 f'(k)&=R(k), &
 f''(k)&=-\frac{m'(k)}{\beta k}<0, &
 f(k)-kf'(k)&=\frac{m(k)}{\beta}>0.
\end{align}
Then $R'(k)=f''(k)<0$. Let us show $R(k)>0$ on $(0,\infty)$. To see this, direct integration on
$[\kappa_2,\kappa_3]$ gives
\begin{equation*}
\label{eq:bubbly-min-Rk3}
 R(\kappa_3)
 =\frac{127}{380}+\frac{137}{190}\log\left(\frac{13}{10}\right)
 >\frac{127}{380}.
\end{equation*}
For $k\in [\kappa_3,\infty)$, since $\beta k\ge (19/20)(13/10)=247/200$, \eqref{eq:bubbly-min-m} implies
\begin{equation*}
 \int_{\kappa_3}^{\infty}\frac{m'(x)}{\beta x}\diff x
 \le \frac{200}{247}\int_{\kappa_3}^\infty m'(x)\diff x=\frac{200}{247}(m(\infty)-m(\kappa_3))=\frac{47}{1235},
\end{equation*}
so
\begin{equation*}
 R(\infty)
 \ge  \frac{127}{380}-\frac{47}{1235}
 =\frac{77}{260}>0.
\end{equation*}
Furthermore, since $m'(k)>0$, we can take $\ubar{m}>0$ such that $m'(k)\ge \ubar{m}$ for $k\in [0,\kappa_2]$. Then \eqref{eq:bubbly-min-Rf} implies
\begin{equation*}
    R(0)\ge \bar{R}+\int_0^{\kappa_2}\frac{\ubar{m}}{\beta x}\diff x=\infty,
\end{equation*}
so $f'(0)=R(0)=\infty$ by \eqref{eq:bubbly-min-f-properties}. Since $f'(\infty)=R(\infty)<R(\kappa_2)=4/5<1=G$, Assumption \ref{asmp:f} holds.

By \eqref{eq:bubbly-min-f-properties}, saving equals $\beta \omega(k)=\beta (f(k)-kf'(k))=m(k)$, so the equilibrium capital equation is
particularly simple:
\begin{equation}
\label{eq:bubbly-min-g}
 k_{t+1}+p_t=m(k_t),
\end{equation}
that is, $g(k,p)=m(k)-p$. It follows from \eqref{eq:bubbly-min-sign} that
\begin{equation*}
\label{eq:bubbly-min-Kstar}
 \cK^*=\set{k>0:g(k,0)=k}
 =\set{\kappa_1,\kappa_2,\kappa_3}
\end{equation*}
and $g(k,0)>k$ for $k\in (0,\kappa_1)$. Since $R$ is strictly decreasing, the largest interest rate at a Diamond steady state is attained at $\kappa_1$. Direct integration gives
\begin{equation*}
\label{eq:bubbly-min-overaccumulation}
    f'(\kappa_1)=R(\kappa_1)=\frac{75}{76}+\frac{137}{190}\log\left(\frac{9}{10}\right)<\frac{75}{76}<1=G.
\end{equation*}
Consequently, \eqref{eq:over-accumulation} holds.

\paragraph{Step 2: a bubbly but asymptotically bubbleless equilibrium.}
Let $h(x)\coloneqq (\bar{R}-x)x$. Then $0<h(x)<\bar{R}x$ for $x\in (0,\bar{R})$. Since $\bar{R}=4/5<1$, for any $x_0\in (0,\bar{R})$, we may recursively define $\set{x_t}_{t=0}^\infty\subset (0,\bar{R})$ by
\begin{equation}
\label{eq:bubbly-min-x}
 x_{t+1}=h(x_t).
\end{equation}
Then $x_{t+1}<\bar{R}x_t<x_t$ and $x_t\to 0$ monotonically. Set $d_0=0$ and define
\begin{align}
\label{eq:bubbly-min-kp}
    k_t&\coloneqq 1+x_t, & p_t&\coloneqq m(k_t)-k_{t+1}, & d_{t+1}&\coloneqq R(k_{t+1})p_t-p_{t+1}.
\end{align}
Then \eqref{eq:bubbly-min-kp} makes the equilibrium system \eqref{eq:system} hold identically.

For $x$ near zero, the middle branch of \eqref{eq:bubbly-min-m} gives
\begin{equation*}
 m(1+x)
 =1+\frac{203}{200}x+\frac{1}{10}x^2-\frac{1}{2}x^3.
\end{equation*}
Using \eqref{eq:bubbly-min-x}, we obtain the exact expression $p_t=q(x_t)$, where
\begin{equation}
\label{eq:bubbly-min-q}
    q(x)\coloneqq \frac{43}{200}x+\frac{11}{10}x^2-\frac{1}{2}x^3.
\end{equation}
Thus $p_t=q(x_t)>0$ for all $t$ if $x_0$ is sufficiently small.  Define
\begin{equation*}
 \delta(x)\coloneqq R(1+h(x))q(x)-q(h(x)).
\end{equation*}
Then $d_{t+1}=\delta(x_t)$ by \eqref{eq:bubbly-min-kp}. A Taylor expansion at zero yields
\begin{equation}
\label{eq:bubbly-min-delta}
 \delta(x)
 =\frac{19687}{95000}x^2+O(x^3).
\end{equation}
Since the leading coefficient is strictly positive, we may choose $x_0$
small enough that $\delta(x)>0$ for every $x\in(0,x_0]$.  Because
$x_t\le x_0$, it follows that $d_{t+1}=\delta(x_t)>0$ for all $t$.

Since \eqref{eq:bubbly-min-x} implies $x_{t+1}/x_t=4/5-x_t\to 4/5$, we have $\lim_{t\to\infty}x_t^{1/t}=4/5$ and $\sum_t x_t<\infty$. Equations \eqref{eq:bubbly-min-q}--\eqref{eq:bubbly-min-delta} imply
\begin{equation}
\label{eq:bubbly-min-asymptotics}
 (p_t,d_{t+1})\sim \left(\frac{43}{200}x_t,\frac{19687}{95000}x_t^2\right).
\end{equation}
Therefore $\sum_t d_t<\infty$, so Assumption~\ref{asmp:D} holds. Because
$G=1$ and hence $D_t=d_t$,
\begin{equation}
\label{eq:bubbly-min-Gd}
 G_d=\lim_{t\to\infty}D_t^{1/t}
 =\left(\frac{4}{5}\right)^2
 =\frac{16}{25}.
\end{equation}
Moreover, \eqref{eq:bubbly-min-asymptotics} implies $d_{t+1}/p_{t+1}=O(x_t)$, so $\sum_{t=1}^\infty d_t/p_t<\infty$ and the equilibrium is bubbly by Lemma \ref{lem:bubble}. On the other hand, $p_t\to 0$, and hence its bubble component satisfies $0<b_t\le p_t\to 0$.  The equilibrium is thus bubbly but
asymptotically bubbleless, with $(k_t,p_t,R_t)\to (1,0,4/5)$. All assumptions of Theorem~\ref{thm:eqset} are satisfied.  Since the equilibrium just constructed is neither a capital collapse equilibrium nor an asymptotically
bubbly equilibrium, case \ref{item:eqset3} of Theorem \ref{thm:eqset} occurs.

\paragraph{Step 3: the constructed equilibrium is the lower endpoint.}
We now prove that $p_0=\min\cP_0$.  Suppose to the contrary that there exists another equilibrium $\set{(k_t',p_t')}_{t\ge 0}$ with $p_0'<p_0$. Proposition \ref{prop:p0}\ref{item:p0_monotone} implies $k_t'>k_t$ and $p_t'<p_t$ for all $t\ge 1$. The constructed equilibrium is asymptotically bubbleless, so $p_0<\max \cP_0$.  Theorem \ref{thm:eqset}\ref{item:eqset3b} therefore implies that $(k_t',p_t')\to (k',0)$ for some $k'\in \cK^*$. Since $k_t'>k_t$ and $k_t\to \kappa_2=1$, it must be $k'\in \set{\kappa_2,\kappa_3}$. The limit cannot be $\kappa_2$. Indeed,
$m'(\kappa_2)=203/200>1$, so there are a neighborhood $I$ of
$\kappa_2$ and a number $\lambda>1$ such that
$m'(k)\ge \lambda$ on $I$.  If both capital paths converged to
$\kappa_2$, they would eventually lie in $I$. Defining $\Delta k_t\coloneqq k_t'-k_t>0$ and $\Delta p_t\coloneqq p_t'-p_t<0$, \eqref{eq:bubbly-min-g} would then imply
\begin{equation*}
    \Delta k_{t+1}=m(k_t')-m(k_t)-\Delta p_t>m(k_t')-m(k_t)\ge \lambda\Delta k_t
\end{equation*}
by the mean value theorem, which is incompatible with $\Delta k_t\to 0$. Consequently, $k'=\kappa_3$. Using \eqref{eq:bubbly-min-f-properties}, \eqref{eq:bubbly-min-Gd}, and $\log(1+x)<x$ for $x>0$, we have
\begin{align*}
 f'(\kappa_3)
 &=R(\kappa_3)
 =\frac{127}{380}
   +\frac{137}{190}\log\left(\frac{13}{10}\right)\\
 &<\frac{127}{380}+\frac{137}{190}\frac{3}{10}
 =\frac{523}{950}
 <\frac{16}{25}=G_d,
\end{align*}
which contradicts $f'(\kappa_3)\ge G_d$ in Proposition \ref{prop:longrun}\ref{item:lr_asymbubbleless}. Therefore $p_0=\min\cP_0$ and the minimum price equilibrium is bubbly but asymptotically bubbleless.
\end{proof}

\begin{rem}
Proposition \ref{prop:bubbly-minimum} shows that Corollary \ref{cor:unique_bubbleless} is sharp:
it identifies the minimum equilibrium as bubbleless conditional on the
existence of a bubbleless equilibrium, but such an equilibrium need not exist
in case \ref{item:eqset3}. The additional present value condition $V_0^*\le p(k_g)$ in Theorem \ref{thm:continuum} therefore
does genuine work by guaranteeing existence of a bubbleless equilibrium. Indeed, the example in Proposition \ref{prop:bubbly-minimum} violates this condition. To see why, by \eqref{eq:bubbly-min-g} the capital equation in the Diamond economy is $k_{t+1}^*=m(k_t^*)$. We have already shown that $m'(1)=203/200>1$, so by \eqref{eq:bubbly-min-sign}, if $k_0^*=k_0>1=\kappa_2$ is sufficiently close to 1, then $k_t^*\to \kappa_3$ by Lemma \ref{lem:diamond}. Since $f'(\kappa_3)<G_d$, the benchmark fundamental value $V_0^*$ in \eqref{eq:bubbleless_cond} is infinite. Furthermore, since we can verify $f'(\kappa_3)<G_d<f'(\kappa_2)<f'(\kappa_1)<G$, condition \eqref{eq:necessity} is violated and Theorem \ref{thm:necessity} does not apply.
\end{rem}

\printbibliography
\end{refsection}
	
\end{document}